\documentclass[12pt]{article}
\usepackage[english]{babel}
\usepackage[latin1]{inputenc}
\usepackage{amsfonts,amssymb,amsmath, epsfig}
\usepackage{color,graphicx,graphics,psfrag}
\usepackage{amsmath,amstext,amssymb,amsfonts, amscd}
\usepackage{hyperref}         
\usepackage{amsthm, cases, subfigure}
\usepackage[percent]{overpic}

\def\Box{\leavevmode\vbox{\hrule
     \hbox{\vrule\kern4pt\vbox{\kern4pt}%
           \vrule}\hrule}}

\newcounter{appendix}
\def\appendix{\advance\c@appendix by 1
   \def\thesection{\Alph{section}}
   \ifnum\c@appendix=1 \setcounter{section}{-1} \fi
   \@startsection {section}{1}{\z@}{-3.5ex plus -1ex minus 
   -.2ex}{2.3ex plus .2ex}{\Large\bf}}

\def\paragraph#1{{\bf #1\ }}

\newtheorem{lemma}{Lemma}[section]  

\newtheorem{theorem}[lemma]{Theorem}

\newtheorem{remark}{Remark}[section]

\title{Self-organized Hydrodynamics in an Annular Domain: Modal Analysis and Nonlinear Effects}

\author{Pierre Degond$^1$ and Hui Yu$^2$}

\date{} 
\begin{document}
\maketitle

\vspace{0.5 cm}
\begin{center}
1. Department of Mathematics, Imperial College London\\
London, SW7 2AZ, United Kingdom\\
pdegond@imperial.ac.uk \\

2. Universit\'{e} de Toulouse; UPS, INSA, UT1, UTM\\
Institut de Math\'{e}matiques de Toulouse, France \\
and CNRS; Institut de Math\'{e}matiques de Toulouse, UMR 5219, France \\
hyu@math.univ-toulouse.fr
\end{center}

\vspace{0.5 cm}
\begin{abstract}
The Self-Organized Hydrodynamics model of collective behavior is studied on an annular domain. A modal analysis of the linearized model around a perfectly polarized steady-state is conducted. It shows that the model has only pure imaginary modes in countable number and is hence stable. Numerical computations of the low-order modes are provided. The fully non-linear model is numerically solved and nonlinear mode-coupling is then analyzed. Finally, the efficiency of the modal decomposition to analyze the complex features of the nonlinear model is demonstrated.
\end{abstract}

\medskip
\noindent
{\bf Acknowledgements:} 
This work was supported by the ANR contract 'MOTIMO' (ANR-11-MONU-009-01). The first author is on leave from CNRS, Institut de Math\'ematiques, Toulouse, France. He acknowledges support from the Royal Society and the Wolfson foundation through a Royal Society Wolfson Research Merit Award and by NSF Grant RNMS11-07444 (KI-Net). The second authors wishes to acknowledge the hospitality of the Department of Mathematics, Imperial College London, where this research was conducted. Both authors wish to thank F. Plourabou\'e (IMFT, Toulouse, France) for enlighting discussions.


\vskip 0.4cm

\section{Introduction}	

Self-organized collective dynamics is ubiquitous in the living world and emerges at all possible scales, from cell assemblies\cite{Yamao_etal_PlosOne11} to animal groups\cite{Tunstrom_etal_PlosCB13}. Collective motion happens when thousands of moving individual entities coordinate with each other through local interactions such as attraction and alignment. As a result, large-scale structures of typical sizes exceeding the inter-individual distances by several orders of magnitude are formed. One of the key questions is to understand how these self-organized structures spontaneously emerge from local interactions without the intervention of any leader. 
With this aim, Individual-Based Models (IBM), i.e. models that describe the behavior of each individual agent have been investigated\cite{Aoki_JapanSocSciFish82,Chate_etal_PRE08,Chuang_etal_PhysicaD07,Couzin_etal_JTB02,Cucker_Smale_IEEETransAutomControl07,Lukeman_etal_PNAS10,Peruani_etal_PRE06,Rejniak_Anderson_SysBiolMed11}. They consist of large systems of ordinary or stochastic differential equations the numerical resolution of which is computationally intensive. To describe large-scale structures coarse-grained models such as Fluid Models (FM) are needed. 
FM describe the dynamics of average quantities such as the mean density or mean velocity of the individuals\cite{Bertozzi_etal_Nonlinearity09,Toner_Tu_PRL95,Toner_etal_AnnPhys05,Topaz_etal_BMB06}. Attempts to derive FM from IBM of collective motion can be found in \cite{Ratushnaya_etal_PhysicaA07}. 
An intermediate step in the hierarchy of models consist of kinetic models (KM)\cite{Bellomo_Soler_M3AS12,Bertin_etal_PRE06,Bertin_etal_JPA09,Bolley_etal_AML12} which are Partial Differential Equations (PDE) describing the evolution of the probability density of the particles in phase-space. FM can be obtained as singular limits of the KM under the hypothesis that the individual scales are much smaller than the system scales. This PDE-based derivation of FM is referred to as the 'Hydrodynamic Limit'.

In \cite{Degond_Motsch_M3AS08}, the hydrodynamic limit of the Vicsek IBM\cite{Vicsek_etal_PRL95} has been performed using an intermediate kinetic description\cite{Bolley_etal_AML12,Degond_Motsch_M3AS08}. The Vicsek IBM describes a noisy system of self-propelled particles interacting through local alignment. In \cite{Degond_Motsch_M3AS08}, it has been shown that the absence of conservation laws (such as momentum conservation) resulting from self-propulsion can be overcome by introducing the new ``Generalized Collision Invariant'' concept. The resulting model, referred to as the ``Self-Organized Hydrodynamics (SOH)'' is written:
\begin{eqnarray}
&&\partial_t \rho + c_1\nabla\cdot(\rho\Omega) = 0,\label{model08-1}\\
&& \rho\left[\partial_t\Omega + c_2(\Omega\cdot\nabla)\Omega \right] + \Theta \mathcal{P}_{\Omega^\perp}\nabla\rho 
= 0,\label{model08-2}\\
&& |\Omega| = 1\label{model08-3},
\end{eqnarray}
where $\rho(x,t)\geq 0$ and $\Omega(x,t)\in \mathbb{R}^d$ are the density and the orientation of the mean velocity of the particles, $c_1>0$, $c_2 \in {\mathbb R}$ and $\Theta>0$ are given parameters, and $d$ is the spatial dimension. We let $\mathcal{P}_{\Omega^\perp} = \mbox{Id} - \Omega \otimes \Omega$ be the projection matrix onto the plane orthogonal to $\Omega$.

This model resembles the usual isothermal gas dynamics equations. Eq.~(\ref{model08-1}) is the continuity equation expressing the conservation of mass. Eq.~(\ref{model08-2}) describes how the velocity orientation evolves under transport by the flow (the second term) and the pressure gradient (the third term, where $\Theta$ is related to the noise in the underlying IBM and has the interpretation of a temperature). However, there are important differences, which arise from the fact the $\Omega$ is not a true velocity but the velocity direction, i.e. it is a vector of unit norm (which is expressed by (\ref{model08-3})). To preserve this geometrical constraint, the pressure gradient has to be projected onto the normal to $\Omega$, which is the reason for the presence of $\mathcal{P}_{\Omega^\perp}$. Other differences stem from the allowed discrepancy between the two constants $c_1$ and $c_2$. While $c_1$ fixes the material velocity to $c_1 \Omega$, the constant $c_2$ describes how $\Omega$ is transported. This discrepancy originates from the lack of Galilean invariance of the underlying IBM, itself resulting from self-propulsion\cite{Toner_etal_PRL98}. This model has been extended into several directions\cite{Degond_Liu_M3AS12,Degond_etal_MAA13,Degond_etal_M3AS14,Degond_etal_arXiv:1404.4886,Degond_Yang_M3AS10} and a rigorous existence result is established in \cite{Degond_etal_MAA13}.

This paper is devoted to the study of the SOH model in an annular domain. Annular geometries allow for simple observations of symmetry-breaking transitions induced by collective motion. When a transition from disordered to collective motion occurs, the system is set into a collective rotation in either clockwise or counter-clockwise directions. Annular geometries are a traditional design for salmon cages in sea farms\cite{Fore_etal_Aquaculture09,Johansson_etal_PlosOne14} and for experiments with locusts\cite{Buhl_etal_Science06,Erban_Haskovec_KRM12}, pedestrians\cite{Moussaid_etal_PlosCB12} or sperm-cell dynamics\cite{Creppy_etal_brevet}. In all these examples, a polarized motion in one direction is observed. In the sperm-cell experiments, the observation of turbulent structures that superimpose to collective rotation motivates the present work. In pure semen, sperm-cells are mostly interacting through volume exclusion. But volume exclusion interactions of rod-like self-propelled particles result in alignment\cite{Peruani_etal_PRE06}. This legitimates the use of the Vicsek model\cite{Vicsek_etal_PRL95} and of its fluid counterpart, the SOH Model\cite{Degond_Motsch_M3AS08}, as models of collective sperm-cell dynamics. The Vicsek model in annular geometry has been shown to exhibit polarized motion in \cite{Czirok_Vicsek_PhysicaA00}. Here, we focus on the SOH model and study its normal modes in annular geometry in both the linear and nonlinear regimes.

We first study the linear modes of the SOH model around a perfectly polarized steady-state in Sec.~\ref{SecSS}. One of the main results of this paper is that these modes are pure imaginary (and thus, stable) and form a countable set. In Sec.~\ref{SecSSNum}, we compute the eigenmodes and eigenfunctions numerically and investigate how the eigenmodes depend on the geometry of the annulus and on the parameters of the model. We then turn towards the nonlinear model with the aims of (i) validating the linear analysis for small perturbations, (ii) investigating how the nonlinearity of the model affects the modal decomposition of the solution and (iii) demonstrating the capabilities of the modal decomposition to analyze the complex features of the nonlinear model. In future work, the modal decomposition will be used to calibrate the model coefficients against experimental data. We first develop the scheme in Sec.~\ref{Secrelax}  and then compare the results for the linear and nonlinear models in Sec.~\ref{SecNonlinearNum}. Finally we draw conclusions and perspecives in Sec.~\ref{sec:conclu}.

\section{Linear Modes of the SOH Model in Polar Coordinates}
\label{SecSS}
\subsection{The SOH model in polar coordinates and perfectly polarized steady-states}

Consider the SOH model (\ref{model08-1})-(\ref{model08-3}) in a two-dimensional annular domain ${\mathcal D} = \{x\in \mathbb{R}^2 \, \, | \, \,  |x| \in (R_1, R_2)\}$.
We introduce polar coordinates $(r,\theta) \in (R_1,R_2)\times[0,2\pi]$ where $r = |x|$ and $\theta$ is the angle between $x$ and a reference direction. We denote by $(e_r,e_\theta)$ the local basis associated to polar coordinates, i.e. $e_r = x/|x| = (\cos \theta, \sin \theta)$ and $e_\theta = e_r^\bot = (-\sin \theta, \cos \theta)$ where the exponent $\bot$ indicates a rotation by an angle $+ \pi/2$. Then, we let $\rho = \rho(r,\theta,t)$ and $\Omega = \Omega(r,\theta,t) = \cos \phi(r,\theta,t) \, e_r + \sin \phi(r, \theta,t) \, e_\theta$, where $\phi(r, \theta, t)$ represents the angle between $e_r$ and $\Omega$. We recall that the constants $c_1$, $c_2$ and $\Theta$ are such that $ c_1 >0, c_2 \in {\mathbb R}, \Theta >0$. For notational convenience, we introduce
\begin{equation} 
\alpha = \frac{c_2}{\Theta} , 
\label{eq:def_alpha}
\end{equation}
and we note that $\alpha$ is of the same sign as $c_2$ and that $c_2/\alpha = 1/\Theta >0$. 
After easy algebra, the SOH model (\ref{model08-1})-(\ref{model08-3}) is equivalent to the following system for $\rho(r,\theta,t)$ and $\phi(r,\theta,t)$ with $(r, \theta) \in (R_1,R_2)\times[0,2\pi]$ and $t>0$,
\begin{align}
&\partial_t \rho + \frac{c_1}{r}\left[\frac{\partial}{\partial r}(r\rho\cos\phi) 
+ \frac{\partial}{\partial\theta}(\rho\sin\phi)\right] = 0,\label{sys1}\\
&\rho\left[\partial_t\phi + c_2\left(\cos\phi\frac{\partial\phi}{\partial r}
+\frac{\sin\phi}{r}\frac{\partial\phi}{\partial\theta}+\frac{\sin\phi}{r}\right)\right] 
+ \Theta \left(\frac{\cos\phi}{r}\frac{\partial\rho}{\partial\theta}
-\sin\phi\frac{\partial\rho}{\partial r}\right) 
= 0,\label{sys2}
\end{align}
subject to the boundary conditions
\begin{equation}
\label{sysbdry}
\phi(R_1,\theta,t) = \phi(R_2,\theta,t) = \pm \frac{\pi}{2}, \qquad  \rho \text{ and } \phi \text{ periodic in }\theta.
\end{equation}
The first boundary condition (\ref{sysbdry}) imposes a tangential flow to the boundary $\partial {\mathcal D}$ and consequently ensures that there is no mass flow across this boundary.  

Now, we look for perfectly polarized steady states of the above system, i.e.  steady states of the form $(\rho_s, \phi_s)$ where $\rho_s$ is independent of $\theta$ and $\phi_s = -\pi/2$ in the whole domain (We have arbitrarily chosen a rotation in the clockwise direction but of course, the results would be the same, mutatis mutandis, with the opposite choice). We have the 

\begin{lemma}
The perfectly polarized steady-states form a one-para\-meter family of solutions given by 
$$\rho_s(r) = \rho_s^* \, r^{\alpha},\qquad 
\phi_s(r,\theta) = -\frac{\pi}{2},
$$
where $\alpha$ is given by (\ref{eq:def_alpha}) and $\rho_s^* >0$ is any positive constant.
\end{lemma}

\begin{proof}
Inserting $\phi_s = -\frac{\pi}{2}$ into (\ref{sys2}) gives $-\frac{c_2}{r}\rho_s+ \Theta\frac{\partial\rho_s}{\partial r} = 0$.  Therefore, there exists $\rho_s^*>0$ such that $\rho_s(r) = \rho_s^* \, r^{\frac{c_2}{\Theta}} = \rho_s^* \, r^{\alpha}$ .
\end{proof}

\subsection{Linearization about perfectly polarized steady-states}
Next we study the linearization of (\ref{sys1}), (\ref{sys2}) about a perfectly polarized steady-state $(\rho_s, \phi_s)$. Given $\varepsilon>0$, a  linear perturbation $(\tilde \rho, \tilde \phi)$ is given by 
$$\rho = \rho_s + \varepsilon \tilde \rho + {\mathcal O}(\varepsilon^2), \qquad 
\phi= \phi_s + \varepsilon \tilde \phi + {\mathcal O}(\varepsilon^2).
$$
Expanding System (\ref{sys1}), (\ref{sys2}) about $(\rho_s, \phi_s)$ and dropping terms of order $\varepsilon^2$ or higher, we deduce that the 
system satisfied by $(\tilde \rho, \tilde \phi)$ is given by:
\begin{align}\label{epsilonsys}
\frac{\partial}{\partial t}\left(
\begin{array}{c}
\tilde \rho\\ \tilde \phi
\end{array}
\right)
+\mathcal{L}\left(
\begin{array}{c}
\tilde \rho\\ \tilde \phi
\end{array}
\right) = 0,
\end{align}
with
\begin{align}\label{drdth}
\mathcal{L}
= \left(
\begin{array}{cc}
0 & c_1\rho_s\\
\frac{c_2}{\alpha \rho_s} & 0
\end{array}
\right)\frac{\partial}{\partial r}
+\left(
\begin{array}{cc}
-\frac{c_1}{r} & 0 \\
0 & -\frac{c_2}{r}
\end{array}
\right)\frac{\partial}{\partial\theta}
+\left(\begin{array}{cc}
0 & \frac{c_1}{r}(1+\alpha)\rho_s\\
-\frac{c_2}{r\rho_s} & 0 
\end{array}
\right), 
\end{align}
supplemented with the boundary conditions:
\begin{eqnarray}
\label{eq:bc1}
\tilde \phi(R_1,\theta,t) = \tilde \phi(R_2,\theta,t) = 0, 
\qquad
\int_{R_1}^{R_2}\int_0^{2\pi} \tilde \rho r\,drd\theta = 0, 
\end{eqnarray}
and $\tilde \rho$, $\tilde \phi$ periodic in $\theta$. These bounday conditions are inherited from (\ref{sysbdry}). The second Eq. in (\ref{eq:bc1}) is a normalization condition whose physical significance is that we are perturbing the steady-state keeping the total particle mass in the system fixed. 

Looking for solutions $(\tilde \rho, \tilde \phi)$ in separation of variables form:
$$
\tilde \rho(r,\theta,t) = e^{\lambda t}\rho_{\lambda}(r,\theta),\qquad
\tilde \phi(r,\theta,t) = e^{\lambda t}\phi_{\lambda}(r,\theta),
$$
we deduce that $(\rho_\lambda, \phi_\lambda)$ must satisfy the following spectral problem: 
\begin{equation}
({\mathcal L} + \lambda \, {\mathcal I}) \left( \begin{array}{c} \rho_\lambda \\ \phi_\lambda \end{array} \right) = 0 , 
\label{eq:spectL}
\end{equation}
supplemented with the boundary conditions (\ref{eq:bc1}), where ${\mathcal I}$ is the identity matrix. We now consider the decomposition of  $(\rho_{\lambda}, \phi_{\lambda})$ into Fourier series, i.e. 
$$\rho_{\lambda}(r,\theta)=  \sum\limits_{n\in\mathbb{Z}}\hat{\rho}_n(r)e^{in\theta},\qquad
\phi_{\lambda}(r,\theta) = \sum\limits_{n\in\mathbb{Z}}\hat{\phi}_n(r)e^{in\theta},
$$where
\[
\hat{\rho}_n(r) = \frac{1}{2\pi}\int_{0}^{2\pi} \rho_{\lambda}(r,\theta)e^{-in\theta}\,d\theta,
\qquad \hat{\phi}_n(r) = \frac{1}{2\pi}\int_{0}^{2\pi} \phi_{\lambda}(r,\theta)e^{-in\theta}\,d\theta.
\]
Then, $(\hat{\rho}_n, \hat{\phi}_n)$ satisfies the following spectral problem: 
\begin{equation} 
({\mathcal L}_n + \lambda \, {\mathcal I}) \left( \begin{array}{c}\hat{\rho}_n  \\\hat{\phi}_n \end{array} \right) = 0 , 
\label{eq:spect}
\end{equation}
with 
\begin{equation}
{\mathcal L}_n 
= \left(\begin{array}{cc}
0 & c_1\rho_s\\
\frac{c_2}{\alpha \rho_s} & 0
\end{array}\right)
\frac{\partial}{\partial r} 
+
\left(\begin{array}{cc}
- in\frac{c_1}{r} & (1+\alpha)c_1\frac{\rho_s}{r} \\
-\frac{c_2}{r\rho_s} &-in\frac{c_2}{r}
\end{array}\right)
,
\label{eq:spect2}
\end{equation}
supplemented with the boundary conditions:
\begin{equation}
\label{eq:bcspect}
\hat{\phi}_n(R_1) = \hat{\phi}_n(R_2) = 0,
\qquad 
\int_{R_1}^{R_2} \hat \rho_n \, r \, dr = 0. 
\end{equation}
We first study the existence of non-trivial solutions $(\lambda, \hat \rho_n, \hat \phi_n)$ to this spectral problem. 

\subsection{Study of the spectral problem for ${\mathcal L}_n$}

We first prove the following 

\begin{lemma} \label{thm_imaginary}
All the eigenvalues $\lambda$ of (\ref{eq:spect}) are pure imaginary.
\end{lemma}   

\begin{proof} 
We recall that $\rho_s(r) = \rho_s^* \, r^{\alpha}$. We introduce the transformation:
\begin{equation}
\rho_n(r) = r^{-\alpha}\hat{\rho}_n(r), 
\qquad \phi_n(r) = \rho_s^* \, r^{\alpha+1}\hat{\phi}_n(r),
\label{eq:rhonphin}
\end{equation}
and find
\begin{eqnarray}
&& \displaystyle\frac{\partial\rho_n}{\partial r} 
\displaystyle + \frac{\alpha}{r^{\alpha+1}}\left(\frac{\lambda}{c_2}-\frac{in}{r}\right)\phi_n = 0,\label{rnsysrho}\\
&&\displaystyle\frac{\partial\phi_n}{\partial r} 
\displaystyle+ \left(\frac{\lambda}{c_1}-\frac{in}{r}\right)r^{\alpha+1}\rho_n = 0, \label{rnsysphi}
\end{eqnarray}
subject to the boundary conditions:
\begin{equation}
\phi_n(R_1) = \phi_n(R_2) = 0, \qquad \int_{R_1}^{R_1} \rho \, r^{\alpha + 1} \, dr = 0. 
\label{eq:Fourier_bc}
\end{equation}

Let $\lambda = \mu+i\nu$ where $\mu$ denotes the real part of $\lambda$ and $\nu$ its imaginary part. Assume that $\mu \neq 0$. We divide (\ref{rnsysphi}) by $\left(\frac{\lambda}{c_1}-\frac{in}{r}\right)r^{\alpha+1} = \left(\frac{\mu}{c_1} + i \big( \frac{\nu}{c_1}-\frac{n}{r} \big) \right)r^{\alpha+1} \not = 0$ and use the first equation (\ref{rnsysrho}) to get (remembering (\ref{eq:def_alpha})):
\begin{equation}\label{phin2nd}
\frac{\partial}{\partial r}\left[\frac{1}{r^{\alpha+1} \Big( \frac{\mu}{c_1}
+i\left(\frac{\nu}{c_1}-\frac{n}{r}\right) \Big) }\frac{\partial\phi_n}{\partial r}\right]
-\frac{1}{r^{\alpha+1}}\left[\frac{\mu}{\Theta}+i\left(\frac{\nu}{\Theta}-\frac{n \alpha}{r}\right)\right]
\phi_n = 0.
\end{equation}
Multiplying (\ref{phin2nd}) by $\bar \phi_n$ (the complex conjugate of $\phi_n$), integrating with respect to~$r$, using the boundary conditions (\ref{eq:Fourier_bc}) and taking the real part of the so-obtained expression, we get
$$ \int_{R_1}^{R_2} \frac{ \frac{\mu}{c_1} } {r^{\alpha+1} \Big( \big( \frac{\mu}{c_1} \big)^2 + \big( \frac{\nu}{c_1} - \frac{n}{r} \big)^2 \Big) } \, \Big| \frac{\partial \phi_n}{\partial r}(r) \Big|^2 \, dr + \int_{R_1}^{R_2} \frac{1}{r^{\alpha+1}} \, \frac{\mu}{\Theta} \, |\phi_n (r)|^2 \, dr = 0 . $$
Since $\Theta >0$ and $\mu \not = 0$, we have $\phi_n = 0$, which shows that there cannot exist a non-trivial solution of the spectral problem when $\mu \not = 0$. \end{proof}

We now determine the eigenvalues $\lambda = i \nu$, $\nu \in {\mathbb R}$ of (\ref{eq:spect}). 
Dropping the index $n$ for simplicity, we introduce the following transformation:
\begin{equation}
u = \sqrt{\frac{c_2}{c_1\alpha}}r^{\frac{\alpha+1}{2}}\rho, \qquad v = \frac{i}{r^{\frac{\alpha+1}{2}}}\phi.
\label{eq:defuv}
\end{equation}
From (\ref{rnsysrho}), (\ref{rnsysphi}), (\ref{eq:Fourier_bc}), $(u,v)$ satisfies the spectral problem:
\begin{equation}
({\mathcal A}_n + \nu {\mathcal I}) \left(\begin{array}{c} u \\ v \end{array}\right) = 0, 
\label{eq:spectAn}
\end{equation}
with the operator $\mathcal{A}_n$ acting on $(u,v) \in L^2(R_1, R_2)^2$  defined by
\[
\mathcal{A}_n
\left(\begin{array}{c}
u\\
v
\end{array}\right) 
= \left(\begin{array}{c}
\displaystyle\frac{c_1n}{r}u 
- \sqrt{\frac{c_1c_2}{\alpha}}\frac{1}{r^{\frac{\alpha+1}{2}}}\frac{\partial}{\partial r}\left(r^{\frac{\alpha+1}{2}}v\right)\\ 
\displaystyle\sqrt{\frac{c_1c_2}{\alpha}}r^{\frac{\alpha+1}{2}}\frac{\partial}{\partial r}\left(\frac{u}{r^{\frac{\alpha+1}{2}}}\right)
+\frac{c_2n}{r}v
\end{array}\right), 
\]
with domain
\[
D(\mathcal{A}_n) = \Big\{
(u,v) \in H^1(R_1, R_2)^2, \quad v(R_1) = v(R_2) = 0\Big\}.
\]
Note that $\nu\in\mathbb{R}$ and that ${\mathcal A}_n$ has real-valued coefficients. Without loss of generality, we look for real-valued eigenfunctions $(u,v)$. For this operator, we have the

\begin{theorem}\label{thm_basis}
The spectrum of $\mathcal{A}_n$ consists of a countable set of eigenvalues $(\nu_{nm})_{m \in {\mathbb N}}$, $\nu_{nm} \in {\mathbb R}$ associated to a complete orthonormal system of $L^2(R_1, R_2)^2$ of eigenfunctions $(u_{nm},v_{nm})_{m \in {\mathbb N}}$. 
Furthermore, $|\nu_{nm}| \to \infty$ as $m \to \infty$. 
\end{theorem}

\begin{proof} 
We first assume that $n \not = 0$ and
drop the subindex $n$ of $\mathcal{A}_n$ for simplicity. 
We will show that there exists $\eta \in \mathbb{R}$ such that the resolvant $R_{\eta} = (\mathcal{A}+\eta\mathcal{I})^{-1}$ exists and is compact in $L^2(R_1, R_2)^2$. 
For this purpose, we consider $(h,g)\in L^2(R_1, R_2)^2$ and look for a solution $(u_\eta,v_\eta)\in D(\mathcal{A})$ of $(\mathcal{A} + \eta {\mathcal I}) (u_\eta,v_\eta) = (h,g)$, i.e.,
\begin{eqnarray}
&&\displaystyle \Big( 1 + \frac{\eta r}{c_1 n} \Big) \frac{c_1n}{r} u_\eta 
- \sqrt{\frac{c_1c_2}{\alpha}}\frac{1}{r^{\frac{\alpha+1}{2}}}\frac{\partial}{\partial r}\left(r^{\frac{\alpha+1}{2}}v_\eta\right) = h,\label{uva}\\ 
&&\displaystyle\sqrt{\frac{c_1c_2}{\alpha}}r^{\frac{\alpha+1}{2}}\frac{\partial}{\partial r}\left(\frac{u_\eta}{r^{\frac{\alpha+1}{2}}}\right)
+\Big( 1 + \frac{\eta r}{c_2 n} \Big) \frac{c_2n}{r}v  =g.\label{uvb}
\end{eqnarray}
We take $|\eta| < \eta_0:= \frac{c_1 |n|}{R_1}$ in such a way that $1 + \frac{\eta r}{c_1n} >0$, $\forall r \in [R_1,R_2]$. Multiplying (\ref{uva}) by $( 1 + \frac{\eta r}{c_1 n} )^{-1} \, r^{-\frac{\alpha-1}{2}}$, taking its derivative with respect to $r$ and using (\ref{uvb}), we deduce that $v_\eta$ satisfies:
\begin{equation}\label{vtildeh}
-\frac{\partial}{\partial r}\left[\frac{1}{r^{\alpha}\big(1+\frac{\eta r}{c_1n}\big)}\frac{\partial}{\partial r}
\left(r^{\frac{\alpha+1}{2}}v_{\eta}\right)\right] 
-\frac{\alpha n^2}{r^{\frac{\alpha+3}{2}}} \, \big(1+\frac{\eta r}{c_2 n} \big) \, v_{\eta}
=\tilde{h}_\eta, 
\end{equation}
with the boundary conditions $v_\eta(R_1) = v_\eta(R_2) = 0$,
where 
\[
\tilde{h}_\eta = \sqrt{\frac{\alpha}{c_1c_2}}\frac{\partial}{\partial r} 
\left(\frac{h}{r^{\frac{\alpha-1}{2}} \big(1 + \frac{\eta r}{c_1n} \big) }\right)
-\frac{\alpha n}{c_2 \, r^{\frac{\alpha+1}{2}}} \, g. 
\]
Note that $\tilde{h}_\eta \in H^{-1}(R_1, R_2)$. Now the problem consists of showing the existence of a unique weak solution $v_\eta \in H^1_0(R_1, R_2)$ to (\ref{vtildeh}), i.e. to the variational formulation:
\begin{equation}\label{varv}
a^\eta (v_\eta, \tilde v) = \langle \tilde{h}_\eta , r^{\frac{\alpha+1}{2}} \, \tilde v \rangle_{H^{-1},H^1_0}\, \, , \qquad \forall \tilde v \in H^1_0(R_1,R_2),
\end{equation}
with 
\begin{eqnarray}
\label{varveta}
&&a^{\eta}(v_{\eta}, \tilde v) =
\int_{R_1}^{R_2}\frac{1}{r^{\alpha} \big(1 + \frac{\eta r}{c_1n} \big)}
\frac{\partial}{\partial r}\left(r^{\frac{\alpha+1}{2}}v_{\eta}\right)
\frac{\partial}{\partial r}\left(r^{\frac{\alpha+1}{2}}\tilde{v}\right)\,dr\\
&&\hspace{1.8cm} -\int_{R_1}^{R_2}\frac{\alpha n^2}{r} \big( 1 + \frac{\eta r}{c_2n} \big) \, v_{\eta}\tilde{v}\,dr,\nonumber
\end{eqnarray}
and where the brackets at the right-hand side of (\ref{varv}) denote duality between the distribution $\tilde h_\eta \in H^{-1}$ and the function $r^{\frac{\alpha+1}{2}} \, \tilde v \in H^1_0$. We introduce $\sigma > 0$. Choosing $\sigma$ large enough and $|\eta|<\eta_0$ there exists $C, \, C'>0$ such that 
$$\sigma - \frac{\alpha n^2}{r^{\frac{\alpha+3}{2}}} \big( 1 + \frac{\eta r}{c_2n} \big) \geq C >0, \quad \frac{1}{r^{\alpha} \big(1 + \frac{\eta r}{c_1n} \big)} > C' >0, \quad \forall r \in [R_1,R_2]. $$
Therefore, the bilinear form $a_{\sigma}^{\eta}(v,\tilde v) = a^\eta(v,\tilde v) + (v,\tilde v \, r^{\frac{\alpha + 1}{2}})$ where $(\cdot,\cdot)$ is the usual inner product in $L^2(R_1,R_2)$ is coercive on $H^1_0(R_1,R_2)$. Consequently, by the Lax-Milgram theorem, the variational formulation 
$$
a^\eta_\sigma (v, \tilde v) = \langle \ell, r^{\frac{\alpha+1}{2}} \, \tilde v \rangle_{H^{-1},H^1_0}\, \, , \qquad \forall \tilde v \in H^1_0(R_1,R_2),
$$
has a unique solution for any $\ell \in H^{-1}(R_1,R_2)$ and the dependence of $v$ upon $\ell$ is continuous. This defines a continuous linear mapping $\mathcal{T}_{\sigma}^\eta: H^{-1}(R_1,R_2) \to H_0^1(R_1, R_2)$, $\ell \mapsto v$. Then $v_\eta\in H^1_0(R_1,R_2)$ is a solution of (\ref{varv}) if and only if $ v_\eta = \mathcal{T}_{\sigma}^\eta (\tilde{h}+\sigma v_\eta )$, i.e.
\begin{align}\label{fixed}
(\mathcal{I}-\sigma\mathcal{T}_{\sigma}^\eta)v_\eta = \mathcal{T}_{\sigma}^\eta\tilde{h}_\eta.
\end{align}
We note that $\mathcal{T}_{\sigma}^\eta\tilde{h}_\eta \in H_0^1(R_1, R_2) \subset L^2(R_1, R_2)$ and that $T_{\sigma}^\eta$, restricted to $L^2(R_1,R_2)$, is a bounded operator from $L^2(R_1, R_2)$ to $H_0^1(R_1, R_2)$. After composition with the canonical imbedding of $H_0^1(R_1, R_2)$ into $L^2(R_1, R_2)$ (still denoted by $T_{\sigma}^\eta$), $T_{\sigma}^\eta$ is a compact operator on $L^2(R_1, R_2)$. Therefore $\mathcal{I} - \sigma\mathcal{T}_{\sigma}^\eta$ is a Fredholm operator. In addition, $\mathcal{T}_{\sigma}^\eta$ is self-adjoint because the bilinear form $a_{\sigma}^\eta$ is symmetric. Thanks to the Fredholm alternative, we have
\[
\text{Im}(\mathcal{I}-\sigma\mathcal{T}_{\sigma}^\eta) = \text{Ker}(\mathcal{I}-\sigma \mathcal{T}_{\sigma}^\eta)^{\perp},
\]
where Im denotes the range and Ker denotes the null space of an operator. 

Suppose that there exists $\eta$ such that $|\eta|<\eta_0$ and  Ker$(\mathcal{I}-\sigma\mathcal{T}_{\sigma}^\eta) = \{0\}$. Then, $(\mathcal{I}-\sigma\mathcal{T}_{\sigma}^\eta)$ is invertible and there exists a unique solution $v_\eta \in H_0^1(R_1, R_2)$ to (\ref{fixed}), or equivalently, to (\ref{vtildeh}). Defining $u_\eta$ by (\ref{uva}) (remember that we suppose $n \not = 0$), then $u_\eta \in L^2(R_1, R_2)$ since $v_\eta \in H^1(R_1, R_2)$ and $h \in L^2(R_1, R_2)$. But, since $v_\eta$ satisfies  (\ref{vtildeh}) in the distributional sense, $u_\eta$ satisfies (\ref{uvb}) in the distributional sense. From the facts that $v_\eta$ and $g$ both belong to $L^2(R_1, R_2)$, we get that $u_\eta \in H^1(R_1, R_2)$. Therefore, $(u_\eta,v_\eta) \in D({\mathcal A})$ and by (\ref{uva}), (\ref{uvb}), it satisfies $({\mathcal A}+\eta {\mathcal I}) (u_\eta,v_\eta) = (h,g)$. This shows that ${\mathcal A}+\eta {\mathcal I}$ is invertible. Furthermore, since $D({\mathcal A})$ is compactly imbedded into $L^2(R_1,R_2)^2$ and that $R_\eta = ({\mathcal A}+\eta {\mathcal I})^{-1}$ is a continuous linear map from $L^2(R_1,R_2)^2$ into $D({\mathcal A})$, the map  $R_\eta$ is compact as an operator of $L^2(R_1,R_2)^2$. 

To prove that there exists $\eta$ such that $|\eta|<\eta_0$ and Ker$(\mathcal{I}-\sigma\mathcal{T}_{\sigma}^\eta) = \{0\}$, we proceed by contradiction. We suppose that for all such $\eta$ there exists a non-trivial $v_\eta \in H_1^0(R_1,R_2)$ such that $(\mathcal{I}-\sigma\mathcal{T}_{\sigma}^{\eta})v_{\eta} = 0$. Equivalently, Eq.~(\ref{varv}) with right-hand side $\tilde{h}_\eta = 0$ has a non-trivial solution $v_\eta \in H_0^1(R_1,R_2)$, which means that $v_\eta$ is an eigenvector for the eigenvalue $0$ of the variational spectral problem: `` to find $\lambda \in {\mathbb R}$ and $v^\lambda \in H^1_0(R_1,R_2)$, $v^\lambda \not = 0$, such that 
$$
a^\eta (v^\lambda, \tilde v) = \lambda (v^\lambda, \tilde v \, r^{\frac{\alpha + 1}{2}}), \quad \forall \tilde v  \in H^1_0(R_1,R_2). \quad ''
$$ 
From the classical spectral theory of elliptic operators\cite{Brezis_Springer11}, we know that the eigenvalues of this problem are isolated. Furthermore, $0$ is a simple eigenvalue. Indeed, Eq.~(\ref{vtildeh}) is a linear second order differential equation. For a given $\eta$, consider two solutions $v_1$, $v_2$ in $H_1^0(R_1,R_2)$ of (\ref{vtildeh}) associated to $\tilde h_\eta = 0$. The Wronskian $v_1 \partial_r v_2 - v_2 \partial_r v_1$ is zero because both $v_1$ and $v_2$ vanish at the boundaries. Therefore, $v_1$ and $v_2$ are linearly dependent and consequently the dimension of the associated eigenvectors is $1$. We realize that the coefficients of $a^\eta$ given by (\ref{varveta}) are analytic functions of $\eta \in [-\eta_0,\eta_0]$. Then, from classical spectral theory again\cite{Kato_Springer13}, one can define an analytic branch of non-zero solutions $\eta \in [-\eta_0,\eta_0] \to v_\eta \in H_1^0(R_1,R_2)$. Now, from (\ref{varv}) with right-hand side $\tilde{h}_\eta = 0$, it follows that for such $v_\eta$, we have $a^\eta(v_\eta, v_\eta) = 0$. Taking the derivative of this identity with respect to $\eta$ at $\eta = 0$, and using the fact that $a^\eta$ is a symmetric bilinear form, we get:
\begin{eqnarray*}
\Big( \frac{d a^\eta}{d\eta}\Big|_{\eta = 0} \Big) (v_0, v_0) +  2 \, a^0 \Big(v_0, \frac{d v_\eta}{d\eta}\Big|_{\eta = 0} \Big) = 0. 
\end{eqnarray*}
Now, since $v_0$ is a variational solution of (\ref{vtildeh}) for $\eta = 0$ with zero right-hand side, the second term is identically zero. Computing the first term, we get 
\begin{eqnarray*}
- n \bigg( \frac{c_1}{n^2} \int_{R_1}^{R_2} \frac{1}{r^{\alpha-1}} \Big| \frac{\partial}{\partial r} \big( r^{\frac{\alpha + 1}{2}} v_0 \big) \Big|^2 \, dr + \frac{1}{\Theta} \int_{R_1}^{R_2} |v_0|^2
\, dr \bigg) = 0 
\end{eqnarray*}
The quantity inside the parentheses is a nonegative quantity which can only be $0$ if $v_0$ is identically zero, which contridicts the hypothesis that $v_0$ is a non-trivial solution. This shows the contradiction and proves that there exists $\eta\in\mathbb{R}$ small enough such that $\text{Ker}(\mathcal{I} - \sigma\mathcal{T}_{\sigma}^{\eta}) = \{0\}$.

In the case $n=0$, it is an easy matter to see that the above proof can be reproduced or alternately, one can invoke directly the spectral theory of elliptic operators. Details are left to the reader. 

Now, for all $n \in {\mathbb Z}$, there exists $\eta_n \in {\mathbb R}$ such that  $R_n = ({\mathcal A}_n + \eta_n {\mathcal I})^{-1} $ exists, and is a compact self-adjoint operator of $L^2(R_1,R_2)^2$. By the spectral theorem for compact self-adjoint operators, there exists a Hilbert basis $(u_{nm},v_{nm})_{m \geq 0}$ of $L^2(R_1,R_2)^2$ and a sequence $(\tau_{nm})_{m \geq 0}$  of real numbers such that $\tau_{nm} \to 0$ as $m \to \infty$ and such that $(u_{nm},v_{nm})$ is an eigenfunction of $R_n$ associated to the eigenvalue $\tau_{nm}$. Then, $(u_{nm},v_{nm})_{m \geq 0}$ is a Hilbert basis in $L^2(R_1,R_2)^2$ of eigenfunctions of ${\mathcal A}_n$ associated to the sequence of eigenvalues $(\nu_{nm})_{m \geq 0}$ with $\nu_{nm} = \frac{1}{\tau_{nm}} - \eta_n$. We have $|\nu_{nm}| \to \infty$ as $m \to \infty$, which concludes the  proof. \end{proof}

We now come back to the original spectral problem (\ref{eq:spectL}). We define 
\begin{equation} \hat \rho_{nm} = \sqrt{\frac{c_1 \alpha}{c_2}} \, r^{\frac{\alpha-1}{2}} \, u_{nm}, \qquad \hat \psi_{nm} = \frac{1}{\rho_s^*} \, r^{\frac{-(\alpha+1)}{2}} \, v_{nm} ,  
\label{eq:eigenL}
\end{equation}
where $(u_{nm},v_{nm})_{m \geq 0}$ is the Hilbert basis of eigenfunctions of ${\mathcal A}_n$ found at Theorem \ref{thm_basis}, and $(\nu_{mn})_{m \geq 0}$ is the associated sequence of eigenvalues. Thanks to the change of functions (\ref{eq:rhonphin}), (\ref{eq:defuv}), $(\hat \rho_{nm}, i\hat \psi_{nm})$ is a Hilbert basis of eigenfunctions of ${\mathcal L}_n$ and $(i\nu_{mn})_{m \geq 0}$ is the associated sequence of eigenvalues. The system $(\hat \rho_{nm}, i\hat \psi_{nm})$ is orthonormal for the inner product 
\begin{equation}
\left\langle \left( \begin{array}{c} \rho \\ \phi \end{array} \right) , \left( \begin{array}{c} \tilde \rho \\ \tilde \phi \end{array} \right) \right\rangle = \int_{R_1}^{R_2} \Big( \frac{ \Theta}{c_1} \, r^{-(\alpha - 1)} \, \rho(r) \, \overline{\tilde \rho(r)}  +  (\rho_s^*)^2 \, r^{\alpha + 1} \, \phi(r) \, \overline{\tilde \phi(r)} \Big) \, dr.
\label{eq:innerprod1}
\end{equation}
Furthermore, we have the following easy Lemma (whose proof is left to the reader): 
\begin{lemma}
Let $i \nu$ be an eigenvalue of ${\mathcal L}_n$ associated to the eigenvector $(\hat \rho, i\hat \psi)$, then $-i\nu$ is an eigenvalue of ${\mathcal L}_{-n}$ associated to the eigenvector $(\hat \rho, - i\hat \psi)$.
\label{lem:parity}
\end{lemma}

\noindent
As a consequence of this Lemma the eigenvalues for $n=0$ come in opposite pairs and we number them such that 
$ \nu_{0 \, 2m} = - \nu_{0 \, 2m-1}$. Therefore, the sequence of eigenvalues of ${\mathcal L}_0$ is $(- i\nu_{0 \, 2m}, i\nu_{0 \, 2m})_{m \geq 1}$. We note that $0$ is not an eigenvalue of ${\mathcal L}_0$.

\subsection{The spectral problem for ${\mathcal L}$ and resolution of the initial value problem}

\medskip
We now turn to the operator ${\mathcal L}$ defined on $L^2((R_1,R_2)\times (0,2\pi))^2$ by (\ref{drdth}) with domain 
\begin{eqnarray*}
&&\hspace{-0.cm}
D({\mathcal L}) = \{ (\rho,\phi) \in H^1((R_1,R_2)\times (0,2\pi))^2 \, | \, 
\phi (R_1, \theta) = \phi (R_2, \theta) = 0, \, \mbox{ a.e. } \theta \in (0,2\pi), \\
&&\hspace{6cm}
(\rho,\phi)(r,0) = (\rho,\phi)(r,2\pi), \,  \mbox{ a.e. } r \in (R_1,R_2) \}. 
\end{eqnarray*}
Using Theorem \ref{thm_basis} and Lemma \ref{lem:parity}, we can state the following theorem (the proof of which is immediate and left to the reader): 

\begin{theorem}
The spectrum of ${\mathcal L}$, $\mbox{Spec} \,  {\mathcal L}$ is discrete, and consists of 
$$ \mbox{Spec} \,  {\mathcal L} =  \Big( \bigcup_{n \geq 1, m \geq 0} \{i \nu_{nm}, - i \nu_{nm} \} \Big) \bigcup \Big(  \bigcup_{m \geq 1} \{i \nu_{0 \, 2m}, - i \nu_{0 \, 2m}\} \Big) , $$
associated to the following basis of eigenvectors 
\begin{eqnarray*}
&& \Big( \bigcup_{n \geq 1, m \geq 0}  \{(\hat \rho_{nm}, i \hat \psi_{nm} )e^{i n \theta}, (\hat \rho_{nm}, - i \hat \psi_{nm} )e^{- i n \theta} \}  \Big) \,  \bigcup \, \Big( \bigcup_{m \geq 1}  \{(\hat \rho_{0\, 2m}, i \hat \psi_{0\,2m} ), \\
&& \hspace{9.5cm} 
(\hat \rho_{0\, 2m}, - i \hat \psi_{0\,2m})\}  \Big) ,
\end{eqnarray*}
which is a Hilbert basis in $L^2((R_1,R_2)\times (0,2\pi))^2$ for the inner product 
\begin{eqnarray}
\label{eq:innerprod2}
&& \hspace{1cm} \left\langle \hspace{-0.2cm}\left\langle\left( \begin{array}{c} \rho \\ \phi \end{array} \right) , \left( \begin{array}{c} \tilde \rho \\ \tilde \phi \end{array} \right) \right\rangle \hspace{-0.2cm}\right\rangle\\
&& \hspace{0.8cm} = \frac{1}{2 \pi} \int_0^{2 \pi} \int_{R_1}^{R_2} \Big( \frac{\Theta}{c_1} \, r^{-(\alpha - 1)} \, \rho(r,\theta) \, \overline{\tilde \rho(r,\theta)}  +  (\rho_s^*)^2 \, r^{\alpha + 1} \, \phi(r,\theta) \, \overline{\tilde \phi(r,\theta)} \Big) \, dr \, d \theta. \nonumber
\end{eqnarray}
\label{thm:specL}
\end{theorem}

\noindent
From this theorem, we have the immediate 

\begin{theorem}\label{thm_pert_sum}
Let $(\tilde \rho, \tilde \phi)(r,\theta,t)$ be the solution of the linearized model (\ref{epsilonsys}) with initial condition $(\tilde \rho_I, \tilde \phi_I) \in  L^2((R_1,R_2)\times (0,2\pi))^2$. By standard semigroup theory, this solution belongs to $C^0([0,T], L^2((R_1,R_2)\times (0,2\pi))^2 \cap L^2([0,T], D({\mathcal L}))$, for all time horizon $T \in {\mathbb R}$. Additionally, we assume that  $(\tilde \rho_I, \tilde \phi_I)$ is real-valued. Then, $(\tilde \rho, \tilde \phi)(r,\theta,t)$ can be expressed as: 
\begin{eqnarray*}
&&\left( \begin{array}{c} \tilde \rho \\ \tilde \phi \end{array} \right) (r,\theta,t) = \sum_{n \geq 1, \, m \geq 0} k_{nm} \left( \begin{array}{r} \hat \rho_{nm}(r) \, \cos(n \theta + \nu_{nm} t + \varphi_{nm})  \\ - \hat \psi_{nm}(r) \, \sin(n \theta + \nu_{nm} t + \varphi_{nm})  \end{array} \right) \nonumber \\
&&\hspace{5cm}
+ \sum_{m \geq 1} k_{0 \, 2m} \left( \begin{array}{r} \hat \rho_{0 \, 2m}(r) \, \cos(\nu_{0 \, 2m} t + \varphi_{0 \, 2m} )  \\ - \hat \psi_{0 \, 2m}(r) \, \sin(\nu_{0 \, 2m} t + \varphi_{0 \, 2m} )  \end{array} \right), 
\end{eqnarray*}
where the series converges in $L^2((R_1,R_2)\times (0,2\pi))^2$ and where $k_{nm}$ and $\varphi_{nm}$ are given, for all $(n,m)$ with ($n\geq 1$ and $m \geq 0$) or  ($n=0$ and $m \geq 1$), by:
\begin{equation}
\label{eq:defknm}
\left\langle \hspace{-0.2cm}\left\langle\left( \begin{array}{c} \rho_I \\ \phi_I \end{array} \right) , \left( \begin{array}{c} \hat \rho_{nm} e^{i n \theta} \\ i \hat \psi_{nm} e^{i n \theta} \end{array} \right) \right\rangle \hspace{-0.2cm}\right\rangle = \frac{1}{2} \, k_{nm} \, e^{i \, \varphi_{nm}} . 
\end{equation}
\end{theorem}

\begin{remark}
The mode indices $n$ and $m$ are related to the number of oscillations in the azimuthal and radial directions respectively. Below, we will refer to $n$ as the azimuthal mode index and $m$ the radial mode index. 
\end{remark}

\section{Numerical Computation of the Eigenvalues and Eigenfunctions}
\label{SecSSNum}

\subsection{Numerical method}
First, we discuss the special case $n=0$. From (\ref{eq:spect}), (\ref{eq:spect2}) (with $\lambda = i \nu$), the function $\hat \psi_0 = i \hat \phi_0$ is a solution to:
\begin{equation}
r\frac{\partial^2 \hat \psi_0}{\partial r^2}
-(\alpha+1)\frac{\partial \hat \psi_0}{\partial r}
+\frac{\alpha\nu^2}{c_1c_2}r \hat \psi_0 = 0,
\end{equation}
with homogeneous boundary conditions $\hat  \psi_0(R_1) = \hat \psi_0(R_2)=0$.  This a classical Bessel equation. Its solution is found e.g. in \cite{Bowman_Dover58}, p. 117 and is given by: 
\begin{equation}\label{solBowman}
\hat \psi_0(r) = \left\{
\begin{array}{ll}
r^{\frac{\alpha+2}{2}}\left[AJ_{\tilde{n}}(\beta r) + BY_{\tilde{n}}(\beta r)\right], 
&\text{ for integer } \tilde{n};\\
r^{\frac{\alpha+2}{2}}\left[AJ_{\tilde{n}}(\beta r) + BJ_{-\tilde{n}}(\beta r)\right], 
&\text{ for noninteger } \tilde{n}.
\end{array}
\right.
\end{equation}
Here, $J_{\tilde{n}}$ and $Y_{\tilde{n}}$ are the Bessel functions of the first and second kinds respectively, 
\[
\beta = \sqrt{\frac{\alpha\nu^2}{c_1c_2}}, \qquad
\tilde{n}^2 = \frac{(\alpha+2)^2}{4}, 
\]
and $(A,B)$ are determined from the boundary conditions. For the sake of simplicity, we focus on the case where $\tilde{n}$ is not an integer, but the extension of the considerations below to integer $\tilde{n}$ would be straightforward. The boundary conditions lead to a homogeneous linear system of two equations for $(A,B)$. The existence of a non-trivial solution $\hat \psi_0$ requires that the determinant of this system vanishes. This leads to the following relation: 
\begin{equation}\label{bdryn0}
\mathfrak{B}(\nu) = 
J_{\tilde{n}}(\beta R_1)J_{-\tilde{n}}(\beta R_2) -
J_{\tilde{n}}(\beta R_2)J_{-\tilde{n}}(\beta R_1) = 0.
\end{equation}
By finding the zeros of $\mathfrak{B}$, we obtain the eigenvalues $\nu$ and then the corresponding eigenfunctions $\hat \psi_0$.
 
For $n\geq 1$, we introduce the following numerical scheme. Given an integer $N$, we define a uniform meshsize $h = \frac{R_2-R_1}{N}$ on the interval $[R_1, R_2]$ and discretization points 
$R_1 = r_0 < r_{\frac12} < r_1 < \cdots < r_j<r_{j+\frac12} <\cdots <r_N = R_2$,
where
$r_j = R_1+jh$ and 
$r_{j+\frac12} = R_1+\left(j+\frac12\right)h$.
For each $n$, $\{\rho_{j+\frac12}\}$ and $\{\psi_j\}$ denote the numerical approximation of $\hat \rho_n$ and $\hat \psi_n = - i \hat \phi_n$ on grid points $r_{j+\frac12}$'s and $r_j$'s respectively. The numerical scheme is
\begin{subequations}\label{rnsysnum}
\begin{numcases}{}
\displaystyle \frac{c_1 n}{r_{j+\frac12}}\rho_{j+\frac12}
- \frac{c_1}{r_{j+\frac12}^{\alpha+1}}\frac{\psi_{j+1}-\psi_j}{h}  = \nu\rho_{j+\frac12}, &$0\leq j\leq N-1$,\\
\displaystyle \frac{c_2r_j^{\alpha+1}}{\alpha}\frac{\rho_{j+\frac12}
-\rho_{j-\frac12}}{h}+\frac{c_2n}{r_j}\psi_j = \nu\psi_j, &$1\leq j \leq N-1$.
\end{numcases}
\end{subequations}
No boundary condition is imposed on $\rho_j$. Concerning $\psi_j$, we have $\psi_0 = \psi_N = 0$.

\begin{remark}\label{remFDn0}
We can modify the scheme (\ref{rnsysnum}) and use it to compute the solution in the case $n=0$ by adding $\rho_{j+\frac12}$ and $\psi_j$ on both sides of the equations respectively. 
\end{remark}

\subsection{Eigenvalules}

In the numerical tests, we choose a set of parameter values given by
\begin{equation}
c_1 = 0.89307, c_2 = 0.69757, \Theta = 0.2, R_1 = 1.9, R_2 = 2.1.
\label{eq:parval}
\end{equation}
Accuracy tests (not reported here) have demonstrated that the numerical scheme is of order $2$ for any value of $n$. In the case $n=0$, we can illustrate the good accuracy of the scheme by comparing the computed value of the eigenvalue to its analytic expression~(\ref{bdryn0}). The comparision is given in Table \ref{n0nu}. With $N=1280$ mesh points, the scheme mentioned in Remark \ref{remFDn0} gives almost the exact eigenvalues.

\begin{table}[ht]
\caption{The eigenvalues $\nu$ for azimuthal mode $n=0$ and various values of the radial mode index $m$. Comparison between the method using the Bessel functions (formula (\ref{bdryn0})) and the scheme mentioned in Remark \ref{remFDn0} with $N=1280$ mesh points. The parameter values are given by  (\ref{eq:parval}).}
{\begin{tabular}{@{}|c||c|c||c|c||c|c|@{}}
  \hline
$m$ 					 &  1  	 &  2  	  &  3  	&  4  	 &  5  	   &  6     \\\hline 
$\nu$(Bessel)            &-6.6631& 6.6631 &-13.2895 &13.2895 &-19.9240 &19.9240 \\\hline
$\nu$(Finite Difference) &-6.6631& 6.6631 &-13.2895 &13.2895 &-19.9240 &19.9240 \\\hline
\end{tabular}}
\label{n0nu}
\end{table}

Table \ref{nnu} lists the eigenvalues corresponding to the first seven radial modes $m=1, \ldots, 6$, for the first four azimuthal modes $n=1, \ldots, 4$ computed by the numerical scheme (\ref{rnsysnum}). It shows that, starting from $m=1$,  the radial eigenmodes come by conjugate pairs of almost the same absolute value but opposite signs (see Columns (1,2), (3,4) and (5,6) in Table \ref{nnu}). The fact that they are not exactly opposite can be attributed to the breaking of the clockwise-anticlockwise symmetry due to the linearization about a steady-state with definite orientation (here the clockwise rotating steady-state has been chosen). Also, the difference between $c_1$ and $c_2$ plays a role in this discrepancy (see Sec. \ref{subsec_paramvar}) where it is shown that varying $c_2$ may increase it).  

\begin{table}[!ht]
\begin{center}
\caption{The first seven radial modes $m=0, \ldots, 6$ for the first four azimuthal modes $n=1, 2, 3, 4$,  computed by the finite difference method (\ref{rnsysnum}) with $N=400$ mesh points. The parameter values are given by  (\ref{eq:parval}). Starting from $m=1$, the radial eigenmodes appear in conjugate pairs of almost (but not equal) absolute value and opposite signs (compare the pairs $(m_1,m_2) = (1,2)$, $(3,4)$, $(5,6)$).}
{\begin{tabular}{|cc||c||c|c||c|c||c|c|}
  \hline
& $m$&  0  &   1    &  2  &  3  &  4  &  5  &  6   \\ 
$n$ & &    &       &    &    &    &    &     \\\hline  \hline
1 & &0.4452 &-6.2647 &7.0618  &-12.8913 &13.6876 &-19.5256 &20.3217 \\\hline
2 & &0.8905 &-5.8668 &7.4608  &-12.4935 &14.0860 &-19.1277 &20.7199 \\\hline
3 & &1.3357 &-5.4692 &7.8603  &-12.0958 &14.4845 &-18.7299 &21.1182 \\\hline
4 & &1.7810 &-5.0720 &8.2601  &-11.6983 &14.8833 &-18.3323 &21.5167 \\\hline
\end{tabular}}
\end{center}
\label{nnu}
\end{table}

\subsection{Eigenfunctions} For illustration purposes, we plot some of the eigenmodes 
\begin{equation}
\big(\rho_{nm}(r,\theta), - \psi_{nm}(r,\theta) \big) := \big(\hat \rho_{nm}(r) \cos n \theta, - \hat \psi_{nm}(r) \sin n \theta \big).
\label{eq:rhoperturb}
\end{equation}
 For a better interpretation of the results, we plot the perturbation density $\rho_{nm}$ 
and the orientation vector  
\begin{equation}
\Omega_{nm} (r, \theta) = \Big(\cos \big(- \frac{\pi}{2}-  \varepsilon \psi_{nm}(r,\theta) \big) \,, \,\sin \big(-\frac{\pi}{2} - \varepsilon \psi_{nm}(r,\theta) \big)\Big).
\label{eq:Omperturb}
\end{equation}
While the former corresponds to the perturbation only, the latter corresponds to the total solution (steady-state plus perturbation). We use a fairly large value of $\varepsilon$ in order to magnify the influence of the perturbation. Since the chosen annular domain is rather thin, we rescale the plot onto an artificially wider annulus. Again, the chosen set of parameters is given by~(\ref{eq:parval}). 

Fig.~\ref{n0t1} displays the modes $(n,m)= (0,4)$ (Figs. \ref{n0t1} (a, b)) and $(n,m)= (4,1)$ (Figs. \ref{n0t1} (c, d)). The left figures (Figs. \ref{n0t1} (a, c)) show the color-coded values of the density perturbations $\rho_{nm}$~(\ref{eq:rhoperturb}) as functions of the two-dimensional coordinates $(x,y)$ in the annulus. 
The right figures (Figs. \ref{n0t1} (b, d)) provide a representation of the orientation vector field $\Omega_{nm}$ (\ref{eq:Omperturb}). In the case of mode $(n,m)= (0,4)$ (Figs. \ref{n0t1} (a, b)), since $n=0$, the solution does not vary in the $\theta$ direction and the density perturbation $\rho_{nm}$ has two zeros in the $r$ direction.  In the case of mode $(n,m)= (4,1)$ (Figs. \ref{n0t1} (c, d)), the solution displays four periods in the $\theta$ direction and has only one zero of the density perturbation $\rho_{nm}$ in the $r$ direction. 
\begin{figure}[!ht]
\begin{center}
\subfigure[$\rho_{nm}: \, (n,m)=(0,4), \nu_{nm} = 13.29.$]
{\includegraphics[width = 0.49\textwidth]{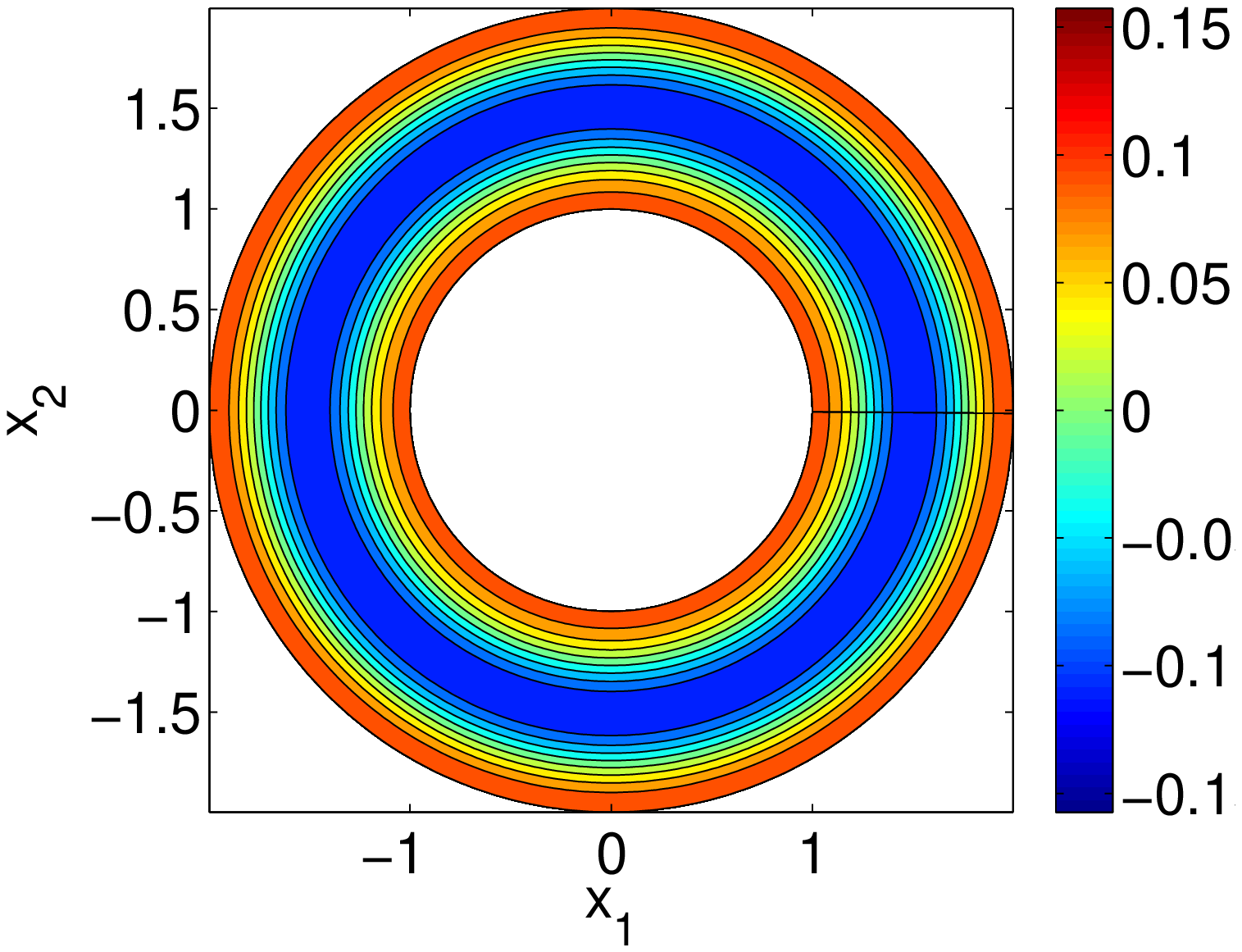}}
\hspace{0.6cm}
\subfigure[$\Omega_{nm}: \,(n,m)=(0,4), \nu_{nm} = 13.29.$]
{\includegraphics[width = 0.42\textwidth]{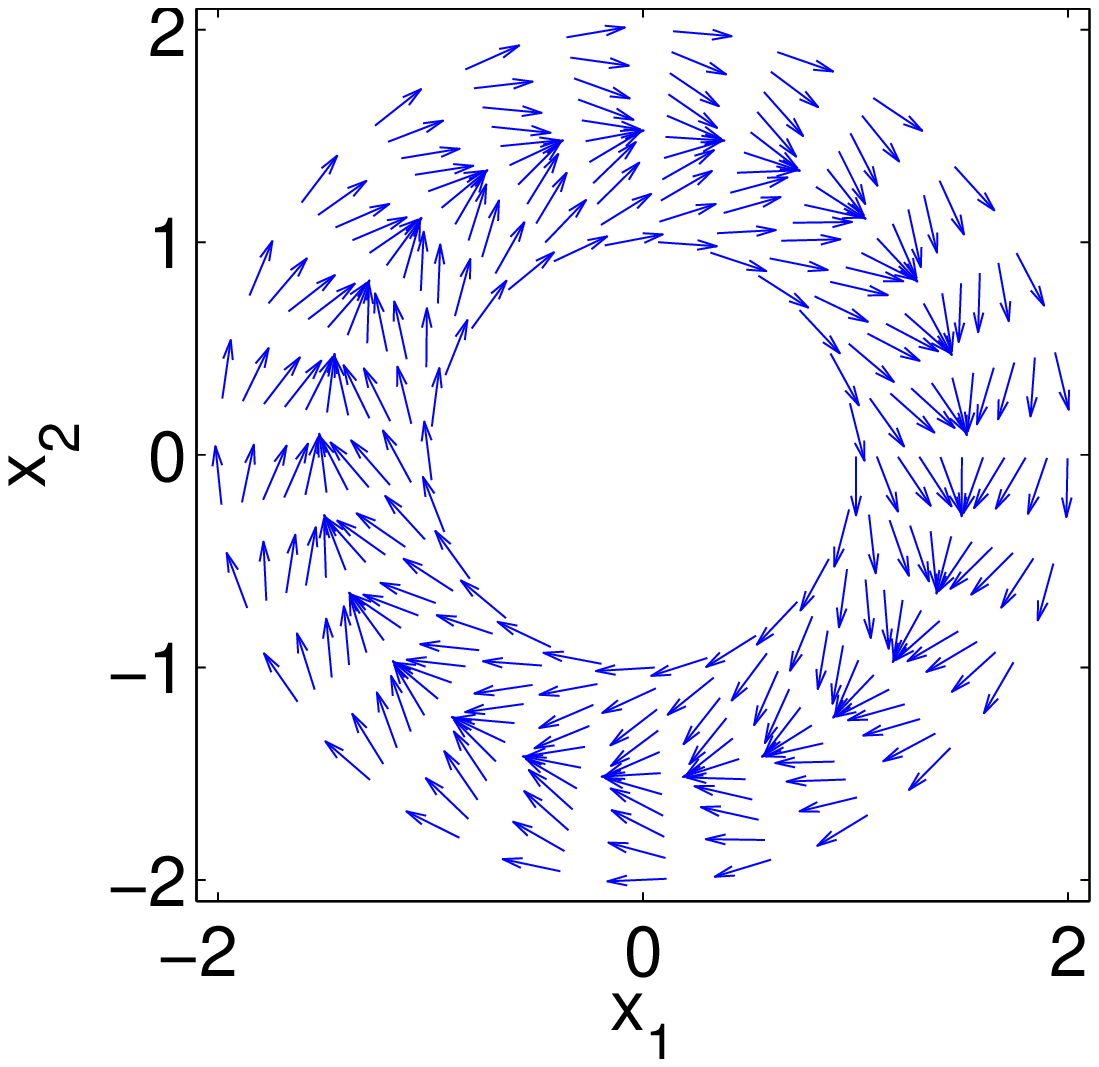}}\\
\vspace{0.4cm}
$\mbox{}$\\
\subfigure[$\rho_{nm}: \, (n,m)=(4,1), \nu_{nm} = -5.07.$]
{\includegraphics[width = 0.49\textwidth]{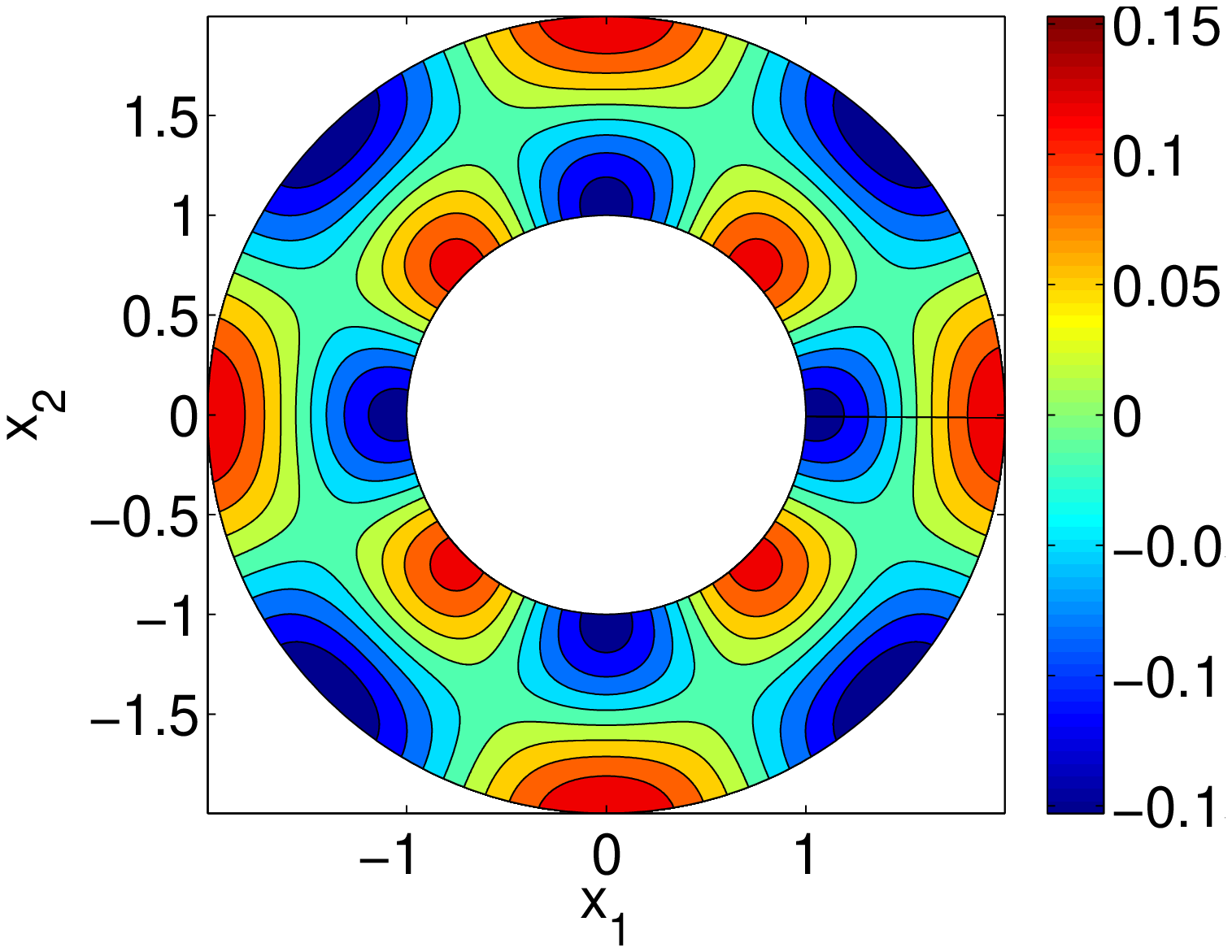}}
\hspace{0.6cm}
\subfigure[$\Omega_{nm}:  (n,m)=(4,1), \nu_{nm} = -5.07.$]
{\includegraphics[width = 0.42\textwidth]{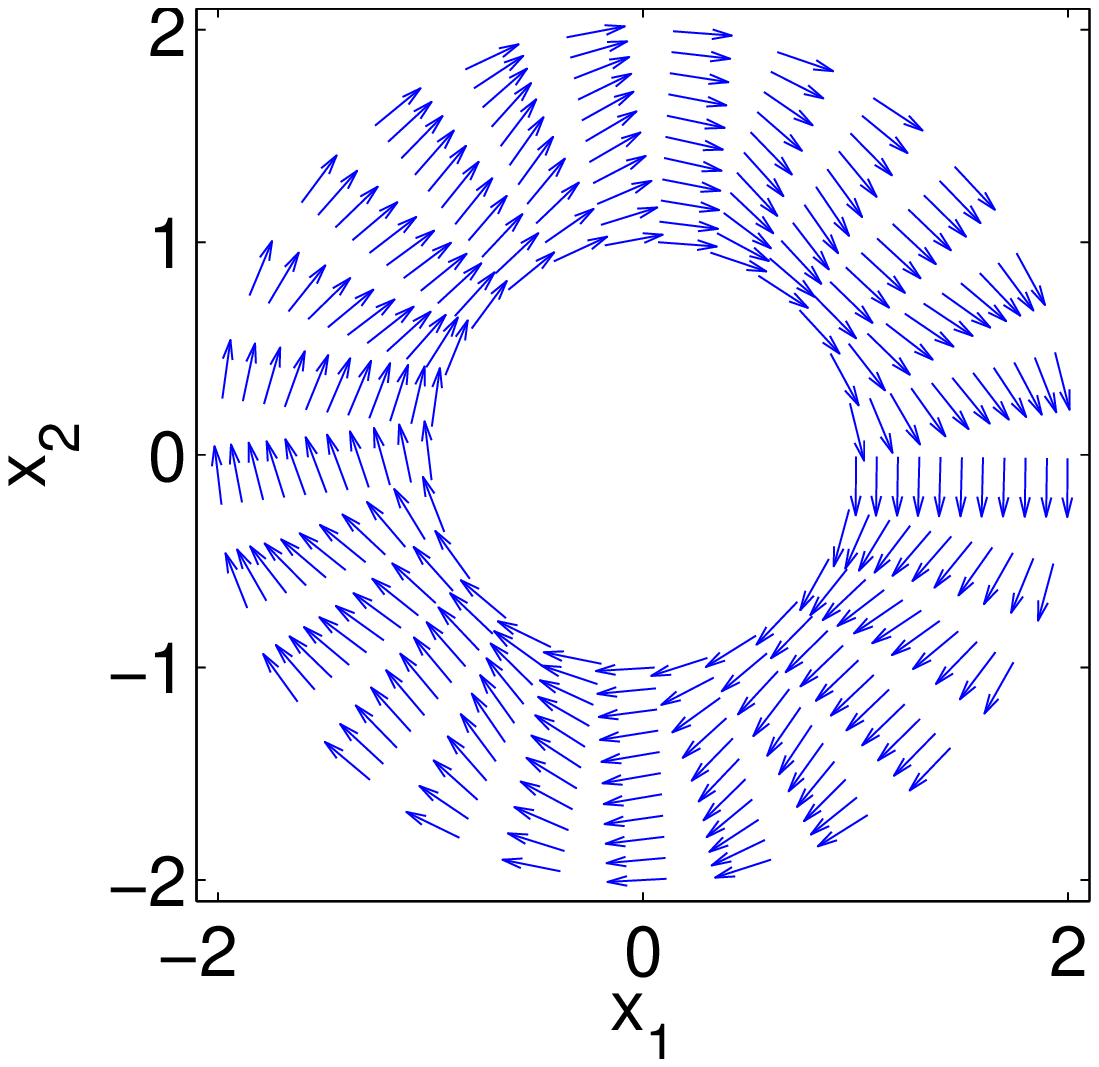}}
\end{center}
\caption{Density perturbation $\rho_{nm}(r,\theta)$  (\ref{eq:rhoperturb}) (Figs. (a, c)) and orientation vector~$\Omega_{nm}$~(\ref{eq:Omperturb}) (Figs.~(b, d)) for mode $(n,m)=(0,4)$ (Figs.~(a, b)) and mode $(n,m)=(4,1)$ (Figs.~(b, d)), as functions of the two-dimensional cartesian coordinates $(x,y) \in {\mathbb R}^2$ in the annulus (color online). The values of the density are color-coded according to the color-bar to the right of the figure. The orientation vector field is represented by blue arrows. The parameter set is given by (\ref{eq:parval}) and $N=400$ mesh points in the radial direction have been used. In Figs.~(a, b), since $n=0$, there is no variation in $\theta$ and the solution is plotted at $t=1$ to make the perturbation visible.}
\label{n0t1}
\end{figure}

\subsection{Variation of the parameters $R_1, R_2, c_1, c_2$ and $\Theta$}
\label{subsec_paramvar}

We numerically investigate the influence of the parameters $R_1, R_2, c_1, c_2$ and $\Theta$ on the eigenvalues. We take the parameter values (\ref{eq:parval}) as references. We vary one of the five parameters $(c_1, c_2, \Theta, R_1, R_2)$ at a time, fixing the other values to those of (\ref{eq:parval}). 

Fig.~\ref{n2para} (a) shows the eigenvalues $\nu$ as functions of the parameters  $R_1$. The inserted-inside Fig.~\ref{n2para} (a)  display how the eigenvalues depend on $c_1$. Fig.~\ref{n2para} (b) shows how the eigenvalues depend on $c_2$. Four eigenvalues corresponding to the modes $n = 2, m=0,1,2,3$ are displayed We observe that when the annular domain becomes narrower, i.e. $R_1$ is larger and closer to $R_2$, the absolute value of $\nu$ is getting larger.  The influence of $R_2$ (not displayed) is similar. As a result, the phase velocities of the modes become faster in a thinner domain, except for $m=0$, which corresponds to no oscillation in the radial direction. As a function of $c_1$ and $\Theta$, $|\nu|$ is monotonically increasing for all values of $n$ (see insert inside Fig.~\ref{n2para} (a) for $c_1$. The behavior as a function of $\Theta$ is similar and not displayed). The effect of a variation of $c_2$ is different: $\nu$ itself (instead of $|\nu|$) is increasing with respect to $c_2$. 
\begin{figure}[!ht]
\begin{center}
\subfigure[Varying $R_1$ (main figure) and $c_1$ (insert)]{
\begin{overpic}[scale=0.47]{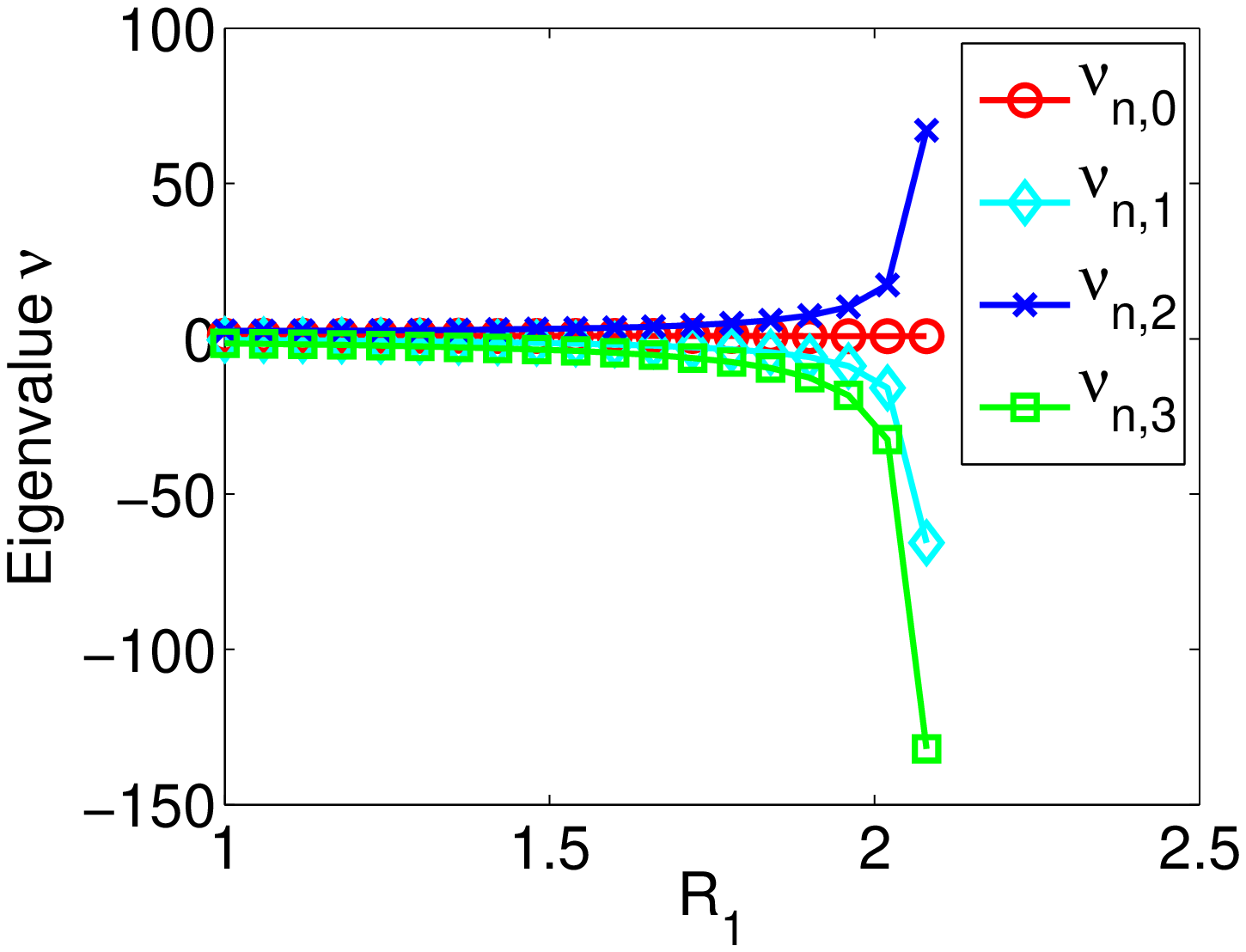}
\put(18,12){\includegraphics[scale=0.17]{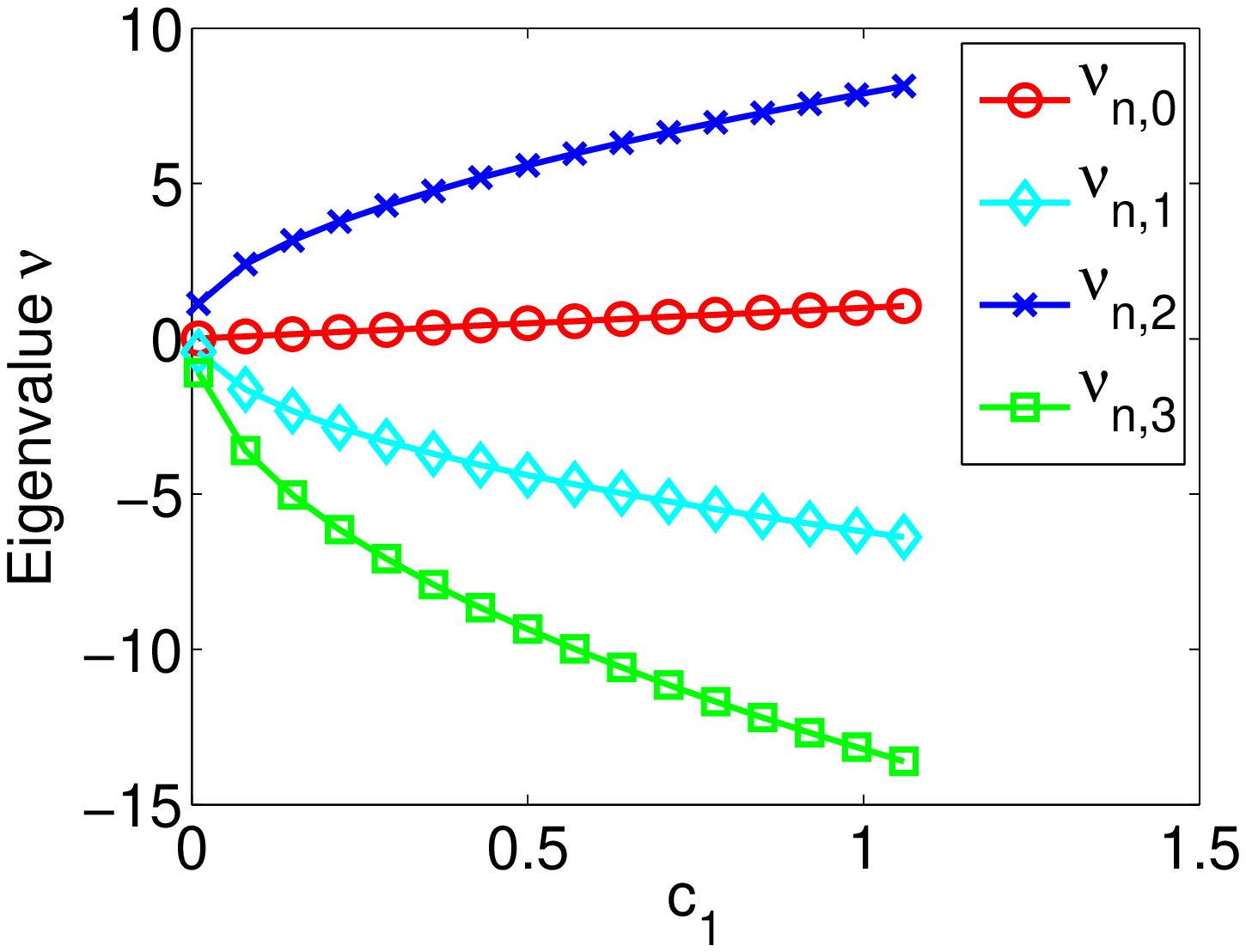}}
\end{overpic}}
\subfigure[Varying $c_2$]{\includegraphics[width = 0.46\textwidth]{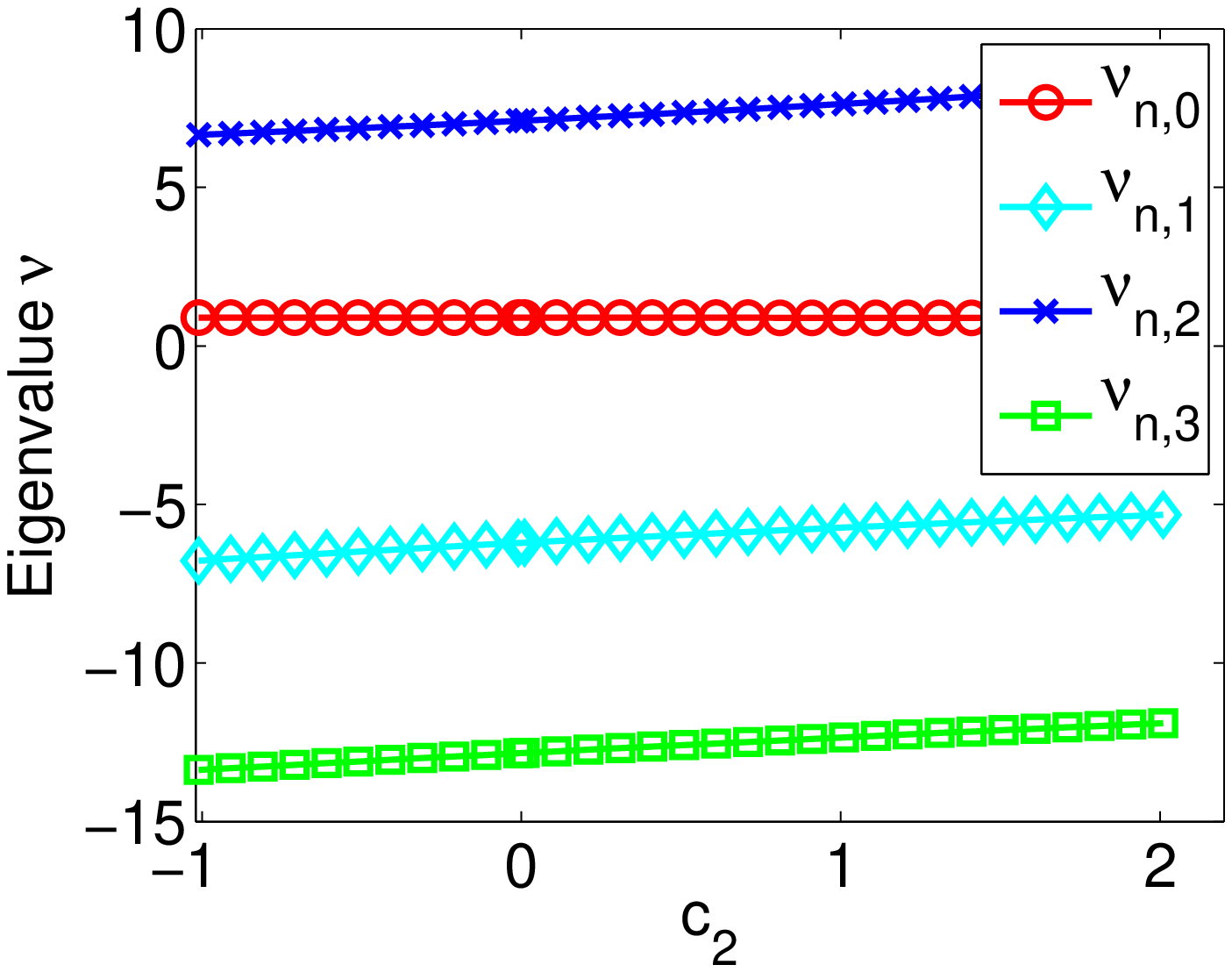}}
\end{center}
\caption{Eigenvalues $\nu$ as functions of the parameters  $R_1$ (Fig. (a)), $c_1$ (insert inside Fig. (a)) and $c_2$ (Fig. (b)). While varying one parameter, the other parameters are fixed to the values given by (\ref{eq:parval}) and $N=400$ mesh points in the radial direction have been used. We display four eigenvalues corresponding to the modes $n = 2, m=0,1,2,3$ }
\label{n2para}
\end{figure}

\section{Numerical Resolution of the Nonlinear SOH Model}
\label{Secrelax}

\subsection{Relaxation model in cylindrical coordinates}

In this section, we discuss the numerical resolution of the nonlinear SOH model (\ref{model08-1})-(\ref{model08-3}), subject to the boundary conditions (\ref{sysbdry}). Its numerical solution will be compared with the solution of the linearized problem found in Sec.~\ref{SecSSNum}. We will further analyze how the nonlinear model departs from its linearization when the perturbation of the steady-state becomes large. One of the difficulties in solving the nonlinear model is the geometric constraint $|\Omega| = 1$ (\ref{model08-3}) and the resulting non-conservativity of the model, arising from the presence of the projection operator ${\mathcal P}_{\Omega^\bot}$ in (\ref{model08-2}). We rely on a method proposed in \cite{Motsch_Navoret_MMS11} where the SOH model is approximated by a relaxation problem consisting of an unconstrained conservative hyperbolic system supplemented with a relaxation operator onto vector fields satisfying the constraint (\ref{model08-3}). In this section, we introduce this relaxation system in cylindrical coordinates in the annular domain.

The relaxation model is given by:
\begin{subequations}\label{relaxmodel}
\begin{numcases}{}
\partial_t \rho^{\eta} + c_1\nabla\cdot(\rho^{\eta}\Omega^{\eta}) = 0, \label{relaxmodel-1}\\
\partial_t(\rho^{\eta}\Omega^{\eta}) + c_2\nabla\cdot(\rho^{\eta}\Omega^{\eta}\otimes\Omega^{\eta})  
+ \Theta \nabla\rho^{\eta} 
= \frac{\rho^{\eta}}{\eta}\left(1-|\Omega^{\eta}|^2\right)\Omega^{\eta}, \label{relaxmodel-2}
\end{numcases}
\end{subequations}
where $\eta \ll 1$ and $\Omega^\eta \in {\mathbb R}^2$ is not constrained to be of unit norm. The relaxation term at the right-hand side of (\ref{relaxmodel-2}) contributes to making $|\Omega^\eta| \approx 1 $. In cylindrical coordinates, let $\Omega^{\eta} = (q^{\eta}\cos\phi^{\eta}, q^{\eta}\sin\phi^{\eta})$, $q^\eta \geq 0$. Dropping the superindex $\eta$ for simplicity, (\ref{relaxmodel}) can be written as
\begin{eqnarray}
&&
\partial_t \rho 
+\frac{c_1}{r} \Big(\frac{\partial}{\partial r}(r\rho q\cos\phi) + \frac{\partial}{\partial\theta}(\rho q\sin\phi) \Big)= 0, 
\label{relaxcylin1} \\
&&
\label{relaxcylin2}
\partial_t  ( \rho q\cos\phi)  +\frac{c_2}{r} \Big( \frac{\partial}{\partial r}(r\rho q^2\cos^2\phi) + \frac{\partial}{\partial\theta}(\rho q^2\sin\phi\cos\phi)-\rho q^2\sin^2\phi \Big)  \\
&&\hspace{6cm}+\Theta  \frac{\partial \rho}{\partial r} = \frac{\rho}{\eta}(1-q^2)q \cos\phi, 
\nonumber\\
&& \label{relaxcylin3}
\partial_t (\rho q\sin\phi) +\frac{c_2}{r} \Big(  \frac{\partial}{\partial r}(r\rho q^2\sin\phi\cos\phi)+\frac{\partial}{\partial\theta}(\rho q^2\sin^2\phi)+\rho q^2\sin\phi\cos\phi \Big) \\
&&\hspace{6cm}+\Theta \frac{1}{r}\frac{\partial\rho}{\partial\theta} = \frac{\rho}{\eta}(1-q^2)q \sin\phi . \nonumber
\end{eqnarray}
Of course, we request that $(\rho,q, \phi)$ are $2 \pi$-periodic with respect to $\theta$.  We supplement the relaxation system with similar boundary conditions as (\ref{sysbdry}). First, we request that the mass flux vanishes on $\partial {\mathcal D}$, implying that
$$ (\rho q \cos \phi (r,\theta,t))|_{r=R_1,R_2}= 0, \quad \forall \theta \in [0, 2\pi], \quad \forall t \in {\mathbb R}_+. $$
When $\eta \ll 1$,  the relaxation term forces $q \approx 1$. Therefore, we assume the same boundary condition (\ref{sysbdry}) as for the SOH model, supplemented with the condition that $q=1$, namely
\begin{equation}
\phi (r,\theta,t))|_{r=R_1,R_2}=\pm \frac{\pi}{2}, \quad q (r,\theta,t))|_{r=R_1,R_2}=1, \quad \forall \theta \in [0, 2\pi], \quad \forall t \in {\mathbb R}_+. 
\label{eq:relbdryphi}
\end{equation}

We have the following theorem, whose proof is analogous to that of Proposition 3.1 in \cite{Motsch_Navoret_MMS11} and is omitted.

\begin{theorem}
The relaxation model (\ref{relaxcylin1})-(\ref{relaxcylin3}) with boundary conditions (\ref{eq:relbdryphi}) converges to the original model (\ref{sys1}), (\ref{sys2}) with boundary conditions (\ref{sysbdry}) as $\eta$ goes to~$0$.
\end{theorem}

\subsection{Relaxation system in conservative form.} The scheme developed in \cite{Motsch_Navoret_MMS11} relies on writing the hyperbolic part of the relaxation system in conservative form. Indeed, the use of a non-conservative form may lead to unphysical solutions, which are not valid approximations of the underlying particle system\cite{Motsch_Navoret_MMS11}. Introducing $(m, u, v)$ defined by
\begin{align*}
m = r\rho,\quad
u = r\rho q\cos(\phi+\theta), \quad
v = r\rho q\sin(\phi+\theta),
\end{align*}
Eqs.~(\ref{relaxcylin1}), (\ref{relaxcylin2}) can be rewritten in terms of the vector function $U = (m, u, v)$ as follows:

\begin{equation}\label{relaxvec}
\partial_t U + \frac{\partial}{\partial r}F(\theta, U) 
+ \frac{\partial}{\partial\theta}G(r,\theta, U) 
= \frac{1}{\eta} H(U), 
\end{equation}
where
\begin{equation}
H(U) = \left(\begin{array}{c} \displaystyle
0\\ \displaystyle
u\left(1-\frac{u^2+v^2}{m^2}\right)\\ \displaystyle
v\left(1-\frac{u^2+v^2}{m^2}\right)
\end{array}\right), \quad 
F(\theta,U) =\left(\begin{array}{c} \displaystyle  \vspace{0.25cm}
c_1(u\cos\theta+v\sin\theta)\\ \displaystyle  \vspace{0.25cm}
c_2\frac{u}{m}(u\cos\theta+v\sin\theta) + \Theta m\cos\theta\\ \displaystyle
c_2\frac{v}{m}(u\cos\theta+v\sin\theta) + \Theta m\sin\theta
\end{array}\right) , 
\end{equation}
and
\begin{equation}
G(r,\theta, U) =\frac{1}{r}\left(\begin{array}{c} \displaystyle \vspace{0.25cm}
c_1(v\cos\theta-u\sin\theta)\\ \displaystyle \vspace{0.25cm}
c_2\frac{u}{m}(v\cos\theta-u\sin\theta) - \Theta m\sin\theta\\ \displaystyle
c_2\frac{v}{m}(v\cos\theta-u\sin\theta) + \Theta m\cos\theta
\end{array}\right) . \vspace{0.2cm}
\end{equation}
Of course, we request that $(m,u,v)$ is $2 \pi$ periodic in $\theta$. The boundary conditions (\ref{eq:relbdryphi}) translate into: 
\begin{eqnarray*}
&&
\qquad (u \cos \theta + v \sin \theta)|_{r=R_1,R_2}=0, \quad
(-u \sin \theta + v \cos \theta)|_{r=R_1,R_2}=\pm 1,
\end{eqnarray*}

\subsection{Numerical method}
\label{subsec_nummet}
We apply the method proposed in \cite{Motsch_Navoret_MMS11}, which consists in splitting (\ref{relaxvec}) into a conservative step and a relaxation step. In the conservative step, we solve (\ref{relaxvec}) with $H=0$. In the relaxation step, we solve (\ref{relaxvec}) with $F=G=0$. When $\eta \ll 1$ this last step can be replaced by a mere normalization of $\Omega$ i.e. changing $(u,v)$ into $\frac{m}{\sqrt{u^2  + v^2}} (u,v)$. The conservative step is solved by classical shock-capturing schemes (see \cite{Motsch_Navoret_MMS11} for details).  We take uniform meshes for $r$ and $\theta$.  Careful accuracy tests (not reported here) have demonstrated that this method is of order $1$. 

\section{Comparison between the Linear and Nonlinear Models}
\label{SecNonlinearNum}

\subsection{Small perturbation}
We take a pure eigenmode as initial condition and compare the numerical solution of the nonlinear model to that of the linearized model. We take an initial condition given by
\begin{equation}
(\rho_I, \phi_I) = (\rho_s, \phi_s) + \varepsilon k_{nm}(\rho_{nm}, - \psi_{nm}), 
\label{eq:inicond}
\end{equation}
with $ (\rho_{nm}, - \psi_{nm})$ given by (\ref{eq:rhoperturb}).  Let $(\rho, \phi)$ denote the exact solution of the nonlinear model (\ref{sys1}), (\ref{sys2}) with boundary conditions (\ref{sysbdry}), $(\rho_\ell, \phi_\ell)$ the solution of the linearized system given by Theorem \ref{thm_pert_sum}, and $(\rho_h, \phi_h)$ the numerical solution of the nonlinear model computed thanks to the method summarized at Sec.~\ref{subsec_nummet}. Consider $\rho$ for example. Formally, we have $\rho- \rho_\ell = {\mathcal O}(\varepsilon^2)$ (we neglect the errors due to the numerical computation of the functions $\hat\rho_{mn}$ which are small), while $\rho- \rho_h= {\mathcal O}(h)$ (since the scheme is of order $1$). Consequently, we have 
\begin{equation}
\rho_h - \rho_\ell =  {\mathcal O}(\varepsilon^2) + {\mathcal O}(h). 
\label{eq:error_linear}
\end{equation}
Fig. \ref{l2mode_multi_n3} (a) shows the $L^1$-distance (below referred to as the ``error'') between the numerical solution of the nonlinear model and that of the linearized system at time $t=0.5$, as a function of the meshsize $h$ for an initial condition (\ref{eq:inicond}) corresponding to mode $(n,m) = (3,2)$ and $k_{nm} = 0.01$. Different perturbation magnitudes $\varepsilon = 0.001$ (red squares), $\varepsilon = 0.0005$ (green triangles), $\varepsilon = 0.0001$ (blue crosses) are used. The parameter values are those of~(\ref{eq:parval}). We notice that for a given value of $\varepsilon$, the error decreases with decreasing values of $h$ until $h$ reaches the approximate values $h = 0.01$ (for $\varepsilon =0.001$ and $\varepsilon = 0.0005$) and $h=0.005$ (for $\varepsilon = 0.0001$). When $h$ is decreased further, the error stays constant but this constant is smaller for smaller $\varepsilon$. This suggests that, consistently with (\ref{eq:error_linear}), the error is dominated by the linearization error for small values of $h$.  This interpretation is also consistent with the observation that the threshold value of $h$ under which the error saturates decreases when $\varepsilon$  becomes smaller. However, the decay of the error seems to be first order in $\varepsilon$ instead of being second order as inferred from (\ref{eq:error_linear}).  This suggests that nonlinear effects are rapidly moving the solution away from the linear regime. However, other diagnostics discussed in the section below show that the linearized model actually provides a very good approximation of the nonlinear model in practical situations.

\subsection{Large perturbations} In this section, we take larger values of $\varepsilon$ and quantify the difference between the solutions of the nonlinear and linearized models. Due to nonlinear mode coupling, it is expected that, even with a pure mode initial condition, new modes will be gradually turned on by the nonlinearity. Let 
$ (\tilde \rho_{h,\varepsilon}, \tilde \phi_{h,\varepsilon}) = \varepsilon^{-1} ( (\rho_h, \phi_h) - (\rho_s, \phi_s) )$ denote the difference between the numerical solution of the nonlinear model and the steady-state, rescaled by the factor $\varepsilon^{-1}$. We define the energy ${\mathcal E}(t)$ of the perturbation as 
$$ {\mathcal E}(t) = \left\langle \hspace{-0.2cm}\left\langle
\left( \begin{array}{c} \tilde \rho_{h,\varepsilon}  \\ \tilde \phi_{h,\varepsilon} \end{array} \right) , 
\left( \begin{array}{c} \tilde \rho_{h,\varepsilon}  \\ \tilde \phi_{h,\varepsilon} \end{array} \right)
\right\rangle \hspace{-0.2cm}\right\rangle = \frac{1}{2} \sum_{n\geq 1, \, m \geq 0} k_{nm}^2(t) + \frac{1}{2} \sum_{m \geq 1} k_{0 \, 2m}^2(t), $$
where the double bracket refers to the inner product (\ref{eq:innerprod2}) and $k_{nm}(t)$ is given by~(\ref{eq:defknm}) with $(\rho_I, \varphi_I)$ replaced by 
$(\tilde \rho_{h,\varepsilon}, \tilde \phi_{h,\varepsilon})$.  The quantity $k_{nm}^2(t)/2$ (respectively $k_{0 \, 2m}^2(t) /2$) represents the energy stored in the modes $(\pm n , m)$ (respectively in the modes $(0,2m-1)$ and $(0,2m)$) at time $t$. In the purely linear case, $k_{nm}(t)$ is independent of $t$. In the nonlinear case, its variation with $t$ provides a measure of how the nonlinearity affects the amplitude of the corresponding modes. 

The initial data is a perturbation of the steady-state by a pure eigenmode, i.e. 
\begin{equation}
\left(\begin{array}{c} \tilde \rho_{h,\varepsilon} \\ \tilde \phi_{h,\varepsilon} \end{array}\right) \bigg|_{t=0}
=  k_{n_0, m_0}
\left(\begin{array}{c}\hat{\rho}_{n_0,m_0}\cos(n_0\theta)\\-\hat{\psi}_{n_0,m_0}\sin(n_0\theta)\end{array}\right), 
\label{eq:modeinitial}
\end{equation}
with $(n_0, m_0) = (3,2)$, $k_{n_0,m_0} = 0.01$. We test different values of $\varepsilon$. For this initial condition, Figs. \ref{l2mode_multi_n3} (b, c) show $k_{nm}(t)$ as a function of $t$ in log-log scale for $\varepsilon = 1$ and $1.5$ respectively. The initial mode $(n_0,m_0)=(3,2)$ is represented with blue X's. In Fig. \ref{l2mode_multi_n3} (b) corresponding to a moderate perturbation $\varepsilon = 1$, only modes  $(n,m)= (0,4)$ (red squares) and $(n,m)=(6,4)$ (purple triangles) appear. Mode $(0,4)$ appears first but saturates while mode $(6,4)$ appears later but reaches higher intensities. Both modes eventually saturate. Likewise, the initial mode decays as higher order modes (not represented in the figure) are turned on by the nonlinearity.  The initial growth of modes $(0,4)$ and $(6,4)$ is linear in log-log scale, which corresponds to a power law growth in time. The two modes have comparable growth rates (the two increasing parts of the curves are parallel straight lines). In the case of a larger perturbation $\varepsilon = 1.5$ displayed in Fig. \ref{l2mode_multi_n3} (c) the situation is strikingly more complex, with a wealth of other modes appearing. In addition to modes $(n,m)=(0,4)$ (red squares) and $(6,4)$ (purple triangle), we notice mode $(6,3)$ (cyan diamonds) and $(3,1)$ (green circles). Mode $(6,3)$ which was absent from Fig. \ref{l2mode_multi_n3} (b) now overtakes mode $(6,4)$ at the beginning, but the latter reaches a higher intensity after some time. The decay of the initial mode $(3,2)$ is also more pronounced. It should be noted that some modes stay extinct all the time. This shows that some pairs of modes are only weakly coupled by the nonlinearity. 

In order to illustrate the successive turn on of the various modes, we have arbitrarily fixed a threshold value $k_t= 0.0005$ (represented by the horizontal dashed blue lines on Figs. \ref{l2mode_multi_n3} (b, c)).  In Fig. \ref{l2mode_multi_n3} (d), we have reported the first time $t_1$ at which $k_{nm}(t)$ reaches the values $k_t$ and plotted it as a function of $\varepsilon$ in log-log scale, for modes $(n,m) = (6,4)$ (blue X's), $(n,m)=(0,4)$ (blue squares) and $(n,m) = (6,3)$ (red circles). The corresponding times $t_1$ are also indicated explicitly on Figs. \ref{l2mode_multi_n3} (b, c)). Fig. \ref{l2mode_multi_n3} (d) shows that for small $\varepsilon$, mode $(6,4)$ is the earliest one to turn on. But as $\varepsilon$ increases, this feature changes and mode $(0,4)$ (which was extinct for smaller value of $\varepsilon$) appears earlier. When $\varepsilon$ is increased further, mode $(6,3)$ also appears, later than $(0,4)$ but earlier than $(6,4)$. This illustrates that the nonlinear mode coupling can exhibit rather complex features and non-monotonic behavior as a function of the perturbation intensity $\varepsilon$. However, even for these large perturbation cases, the amplitude of the initial mode always remains one order of magnitude larger than those of the successively excited modes. This shows that the linear model still provides a fairly good approximation of the solution of the nonlinear model.

\begin{figure}[!ht]
\begin{center}
\subfigure[${\rm dist}_{\rho}$]{\includegraphics[width = 0.49\textwidth]{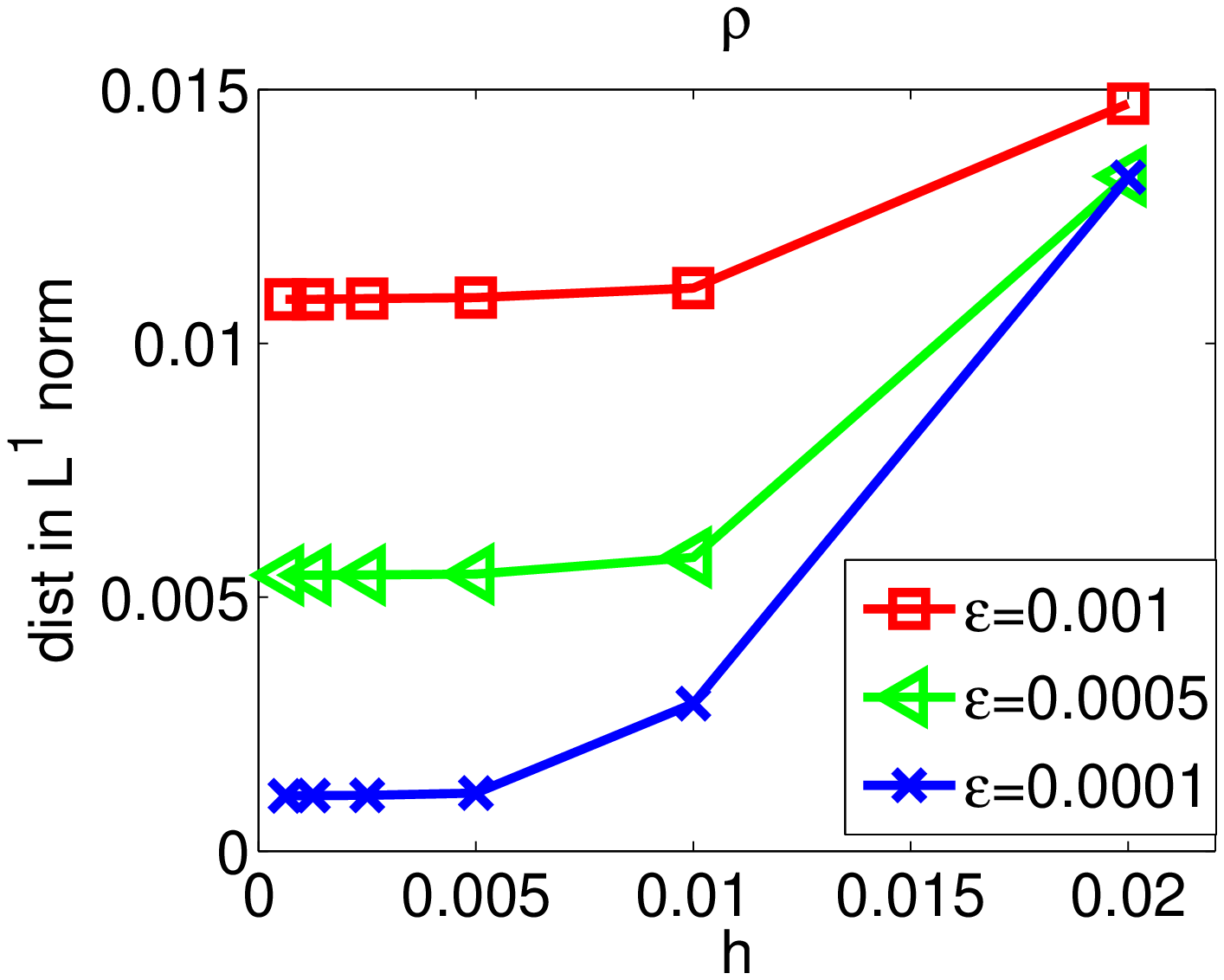}}
\subfigure[$\varepsilon = 1$]{\includegraphics[width = 0.49\textwidth]{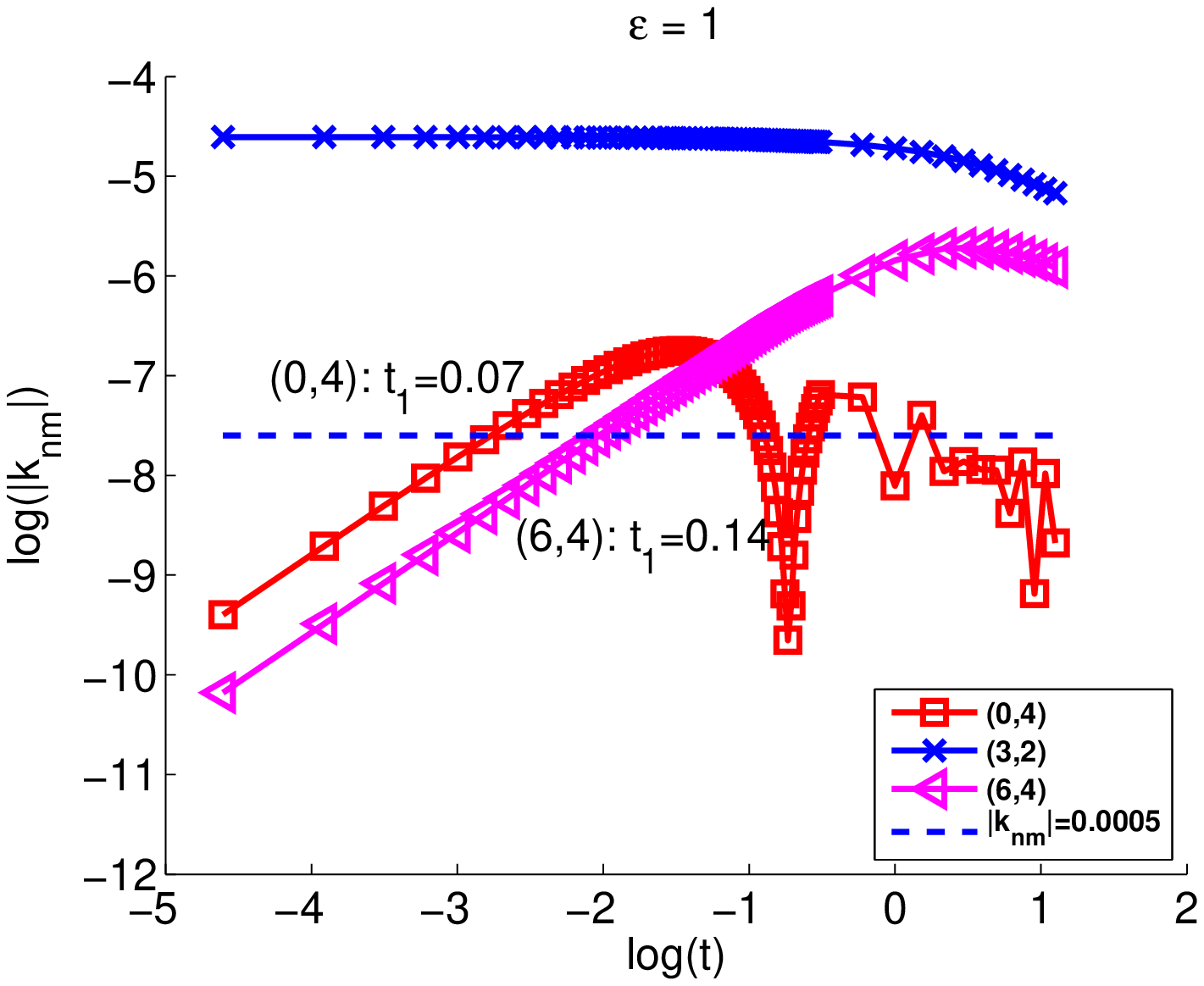}}
\subfigure[$\varepsilon = 1.5$]{\includegraphics[width = 0.49\textwidth]{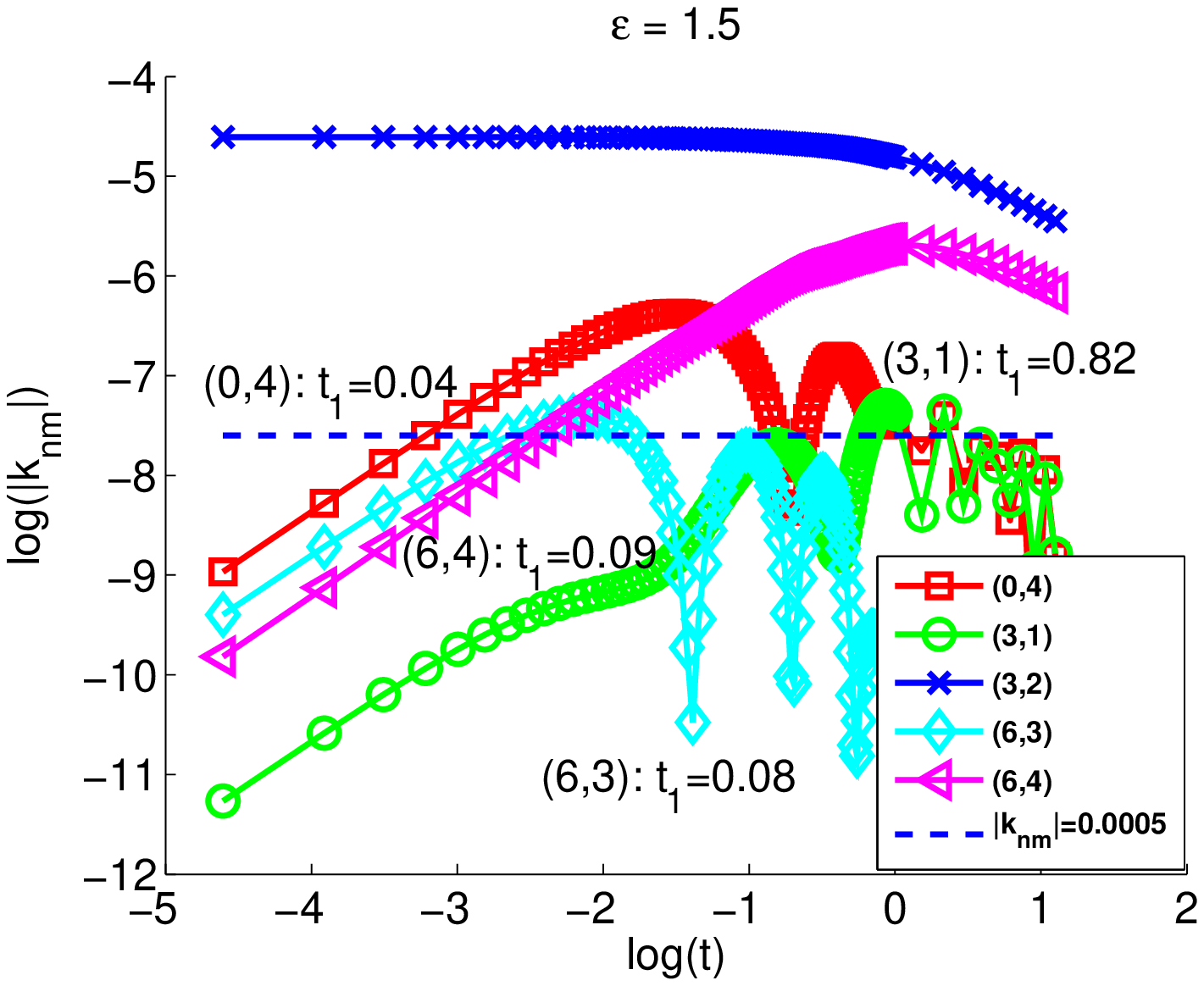}}
\subfigure[$\log(t_1)$ versus $\log(\varepsilon)$]{\includegraphics[width = 0.45\textwidth]{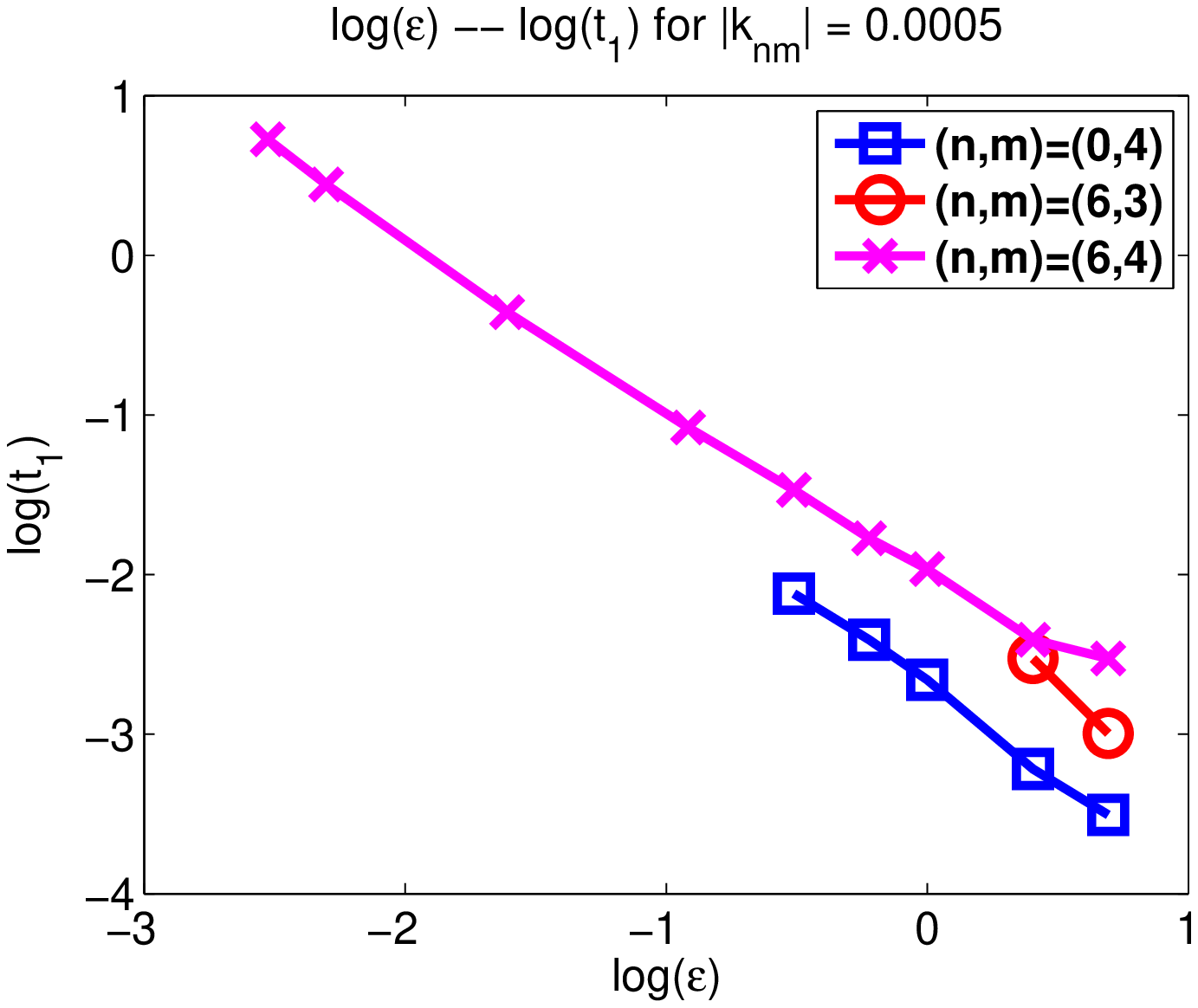}}\\
\end{center}
\caption{
{\bf (a) small perturbation case.} $L^1$-distance between the numerical solution of the nonlinear model and that of the linearized system at time $T=0.5$, as a function of the meshsize $h$ for an initial condition (\ref{eq:inicond}) corresponding to mode $(n,m) = (3,2)$: $\varepsilon = 0.001$ (red squares), $\varepsilon = 0.0005$ (green triangles), $\varepsilon = 0.0001$ (blue crosses). We use $\Delta t = 0.0001$,  $N_r =N_{\theta} = 10$, $20$, $40$, $80$, $160$, $320$ mesh points. The parameter values are those of (\ref{eq:parval}). \hspace{13cm} 
{\bf (b, c, d) large perturbation case.} (b) $k_{nm}$ as a function of $t$ in $\log$-$\log$ scale for $\varepsilon = 1$ and for $(n,m)=(3,2)$ (the initial mode, blue X's), $(0, 4)$ (red squares) and $(6, 4)$ (purple triangles). (c) $k_{nm}$ as a function of $t$ in $\log$-$\log$ scale for $\varepsilon = 1.5$ and for $(n,m)=(3,2)$ (the initial mode, blue X's), $(0, 4)$ (red squares), $(6, 4)$ (purple triangles), $(6, 3)$ (cyan diamonds) and $(3, 1)$ (green circles).  (d): Mode turn-on time $t_1$ as a function of $\varepsilon$ for modes $(n,m) = (6,4)$ (blue X's), $(0,4)$ (blue squares) and $(6,3)$ (red circles). The mode turn-on time $t_1$ is the first time for which $k_{nm}(t)$ reaches the threshold value $k_t$ represented by the horizontal dashed blue line on Figs. (b) and (c). The parameter values are those of (\ref{eq:parval}) and $\Delta t = 0.0005$, $N_r = N_{\theta} = 400$.}
\label{l2mode_multi_n3}
\end{figure}

We now investigate the qualitative features of the solution in a large amplitude case. Figs. \ref{n4nu1t0_2} shows the numerical solution corresponding to a pure mode initial data~(\ref{eq:modeinitial}) with $(n_0,m_0) = (4,1)$, $k_{n_0,m_0} = 1$ and $\varepsilon = 0.01$. It displays the density $\rho$ at times $t=0$ (left) and $t=2$ (right) as a function of the two-dimensional position coordinates $(x,y)$ in the annulus, in color code (color bar to the right of the figure). We observe that the solution remains $\pi/2$-periodic in the $\theta$-direction (as the linear mode would be) but the density contours have lost their sinusoidal shape. Instead, oblique shock waves have formed and are reflected by the boundary. These simulations suggest the existence of unsmooth periodic solutions of the nonlinear SOH model in this geometric configuration. 

\begin{figure}[!ht]
\begin{center}
\subfigure[$t=0$]{\includegraphics[width = 0.49\textwidth]{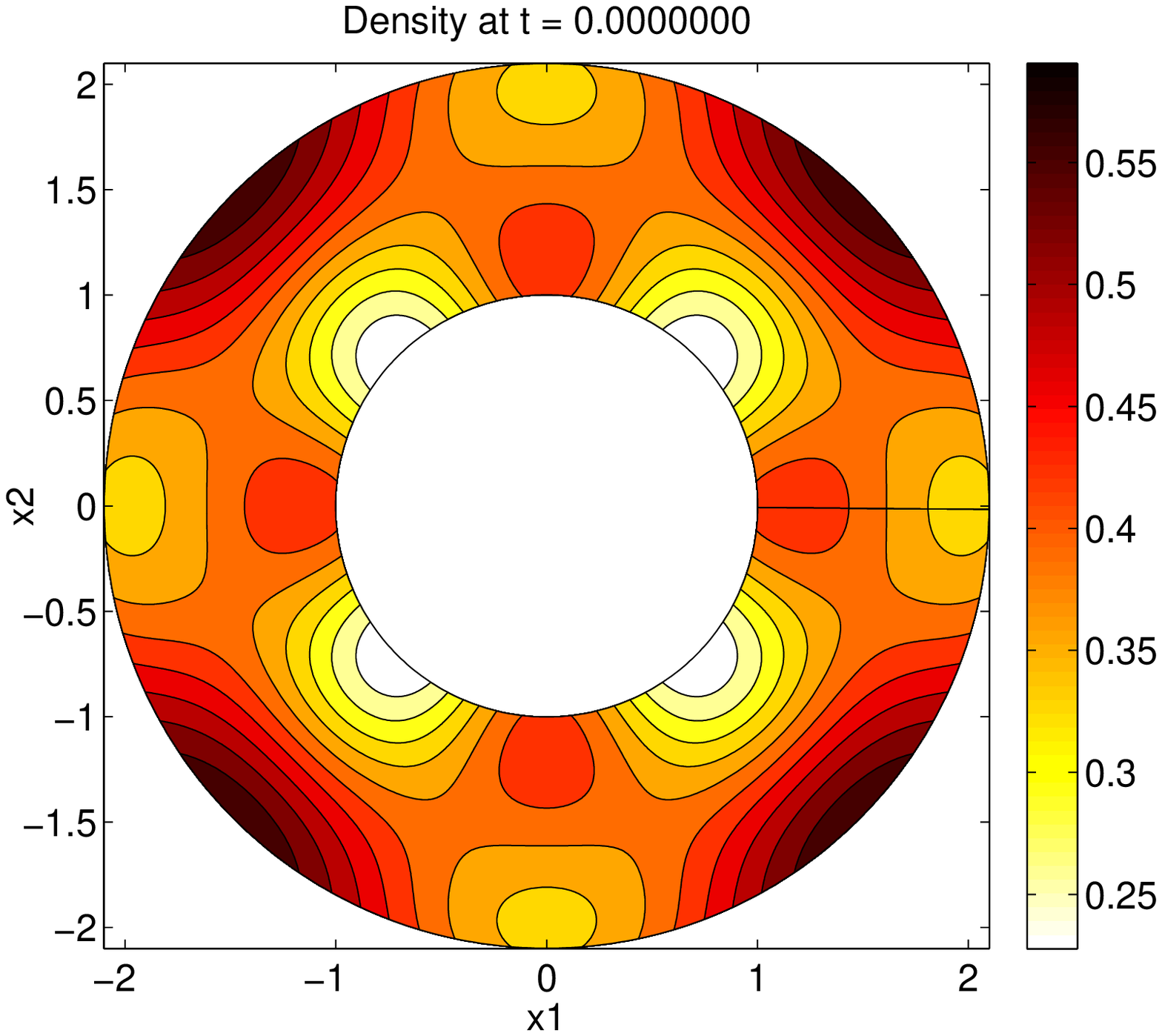}}
\subfigure[$t=2$]{\includegraphics[width = 0.49\textwidth]{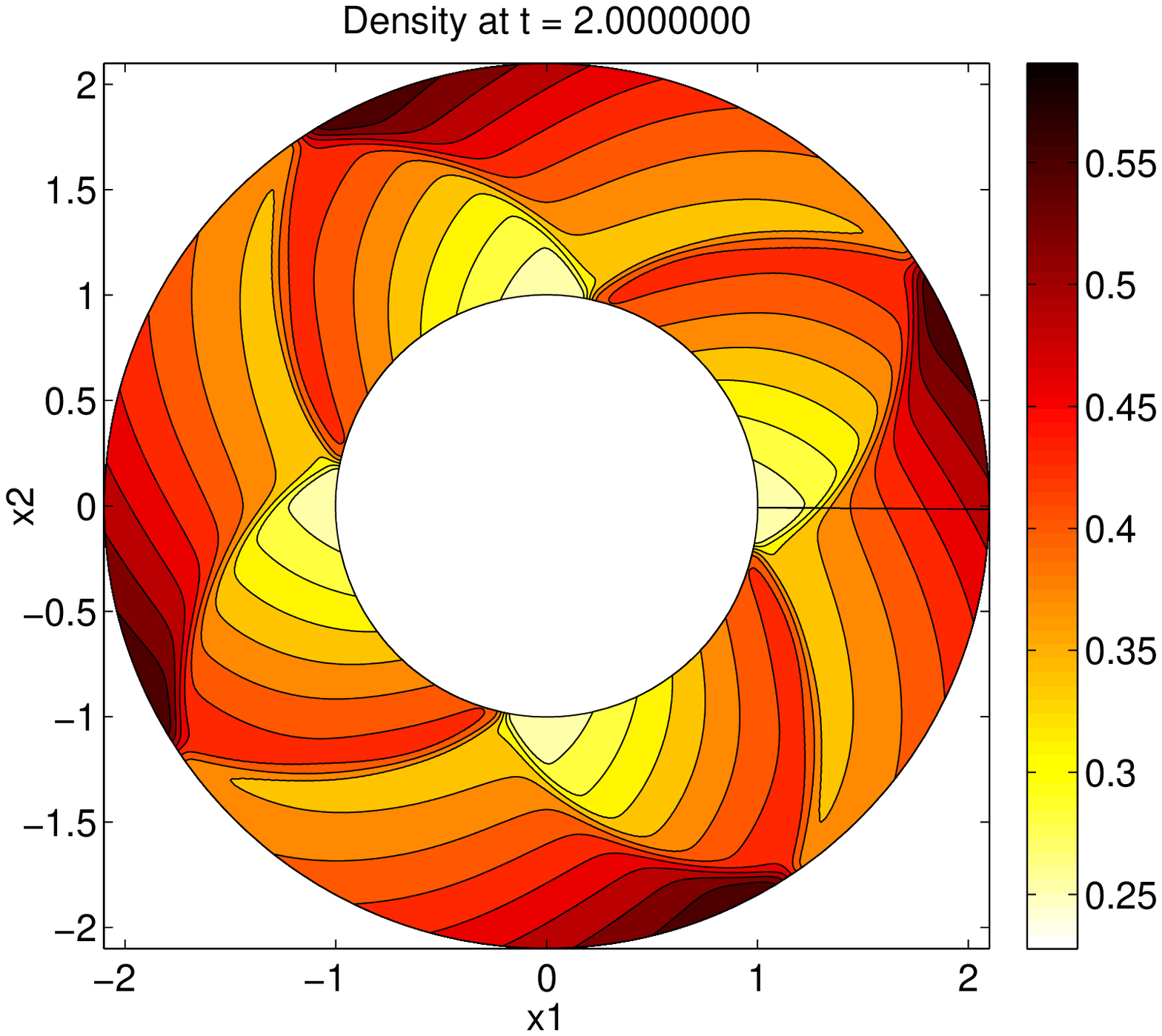}} \\
\end{center}
\caption{Density $\rho$ as a function of the two-dimensional position coordinates $(x,y)$ in the annulus. The initial condition is given by (\ref{eq:modeinitial}) with $(n_0,m_0) = (4,1)$, $k_{n_0,m_0} = 1$ and $\varepsilon = 0.01$. The solution is represented at time $t=0$ (left) and $t=2$ (right). The density is color coded according to the color bar to the right of the figure. The parameters are given by (\ref{eq:parval}), $\Delta t = 0.0005$, and $N=400$ mesh points in both the radial and azimuthal directions have been used.}
\label{n4nu1t0_2}
\end{figure}

We now investigate a large perturbation amplitude case with a random intitial data. More precisely, the initial data is given by a random combination of eigenmodes such that $n \leq 12$ and $m \leq 12$ as follows:
\begin{equation}
\left(\begin{array}{c} \tilde \rho_{h,\varepsilon} \\ \tilde \phi_{h,\varepsilon} \end{array}\right) \bigg|_{t=0}
=  \sum_{\substack{0\leq n, m\leq 12\\ (n,m) \neq (0,0)}}k_{nm}
\left(\begin{array}{c}\hat{\rho}_{nm}\cos(n\theta+\varphi_{nm})\\-\hat{\psi}_{nm}\sin(n\theta+\varphi_{nm})\end{array}\right),
\label{eq:inicondrandom}
\end{equation}
where $k_{nm}$ and $\varphi_{nm}$ are randomly sampled in the intervals $(0,1]$ and $[0,2\pi]$ respectively, according to the uniform distribution. The numerical simulation is performed with $N_r = N_{\theta} = 640$, $\Delta t = 0.0005$ and $\varepsilon = 0.0025$. Fig. \ref{fig:random} shows the numerical solution at time $t=2$ (which approximately corresponds to the rotation of the fluid by a quater of a circle). It displays the density $\rho$ as a function of the two-dimensional position coordinates $(x,y)$ in the annulus, in color code (color bar to the right of the figure). Fig. \ref{fig:random} (a) shows the solution of the linearized model, obtained by summation of the corresponding eigenmodes, while Fig. \ref{fig:random} (b) displays the numerical solution of the nonlinear model with the same initial condition. We observe a very good agreement between the linearized and nonlinear solutions, in spite of a fairly large perturbation amplitude. By looking carefully, one notices that the nonlinear solution has slightly lower maxima and larger minima, due to the action of numerical diffusion (which is absent from the linearized solution). The nonlinear solution also exhibits steeper gradients due to nonlinear shock formation. 

The use of the linearized solution results in considerable computational speed-up compared to that of the nonlinear one. Indeed, the computation of the eigenmodes and their summation to construct the solution is almost instantaneous on a standard laptop. By comparison, the computation of the nonlinear solutions takes of the order of an hour. Therefore, given the considerable computation speed-up, we consider that the performances of the linearized model are excellent. These performances make the linearized model a model of choice to perform parameter calibration on experimental data. Indeed, parameter calibration involves the iterative resolution of a minimization problem which consists of finding the set of parameters which minimize the distance between the solution and the data. With the linearized model, this calibration phase can be expected to require very little computational time. This is important, since this set of parameters is expected to change from one experiment to the next and consequently, the calibration phase must be performed for each experiment. A real-time analysis of an experiment therefore requires a very efficient algorithm.

\begin{figure}[!ht]
\begin{center}
\subfigure[Linear model]{\includegraphics[width = 0.49\textwidth]{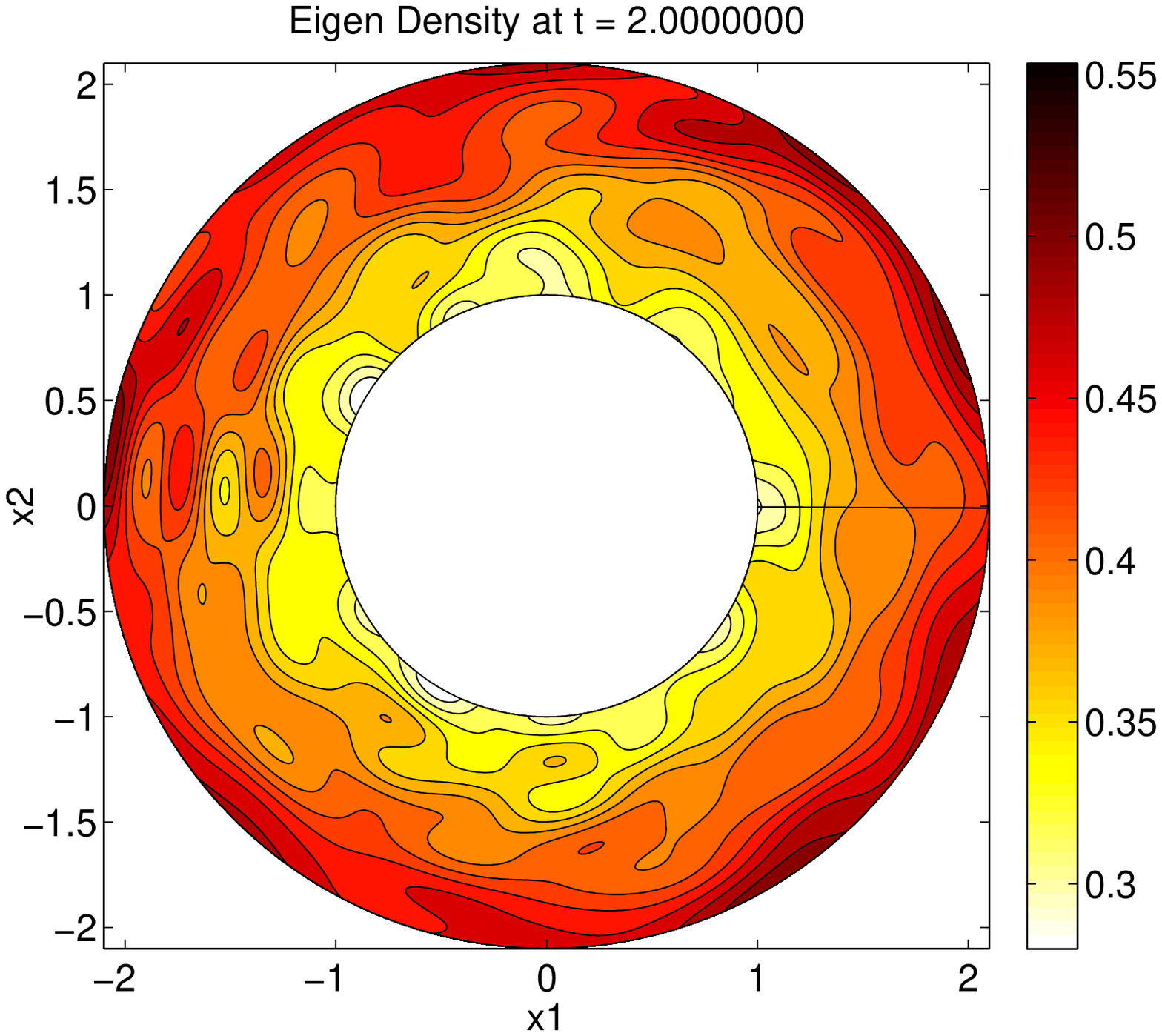}}
\subfigure[Nonlinear model]{\includegraphics[width = 0.49\textwidth]{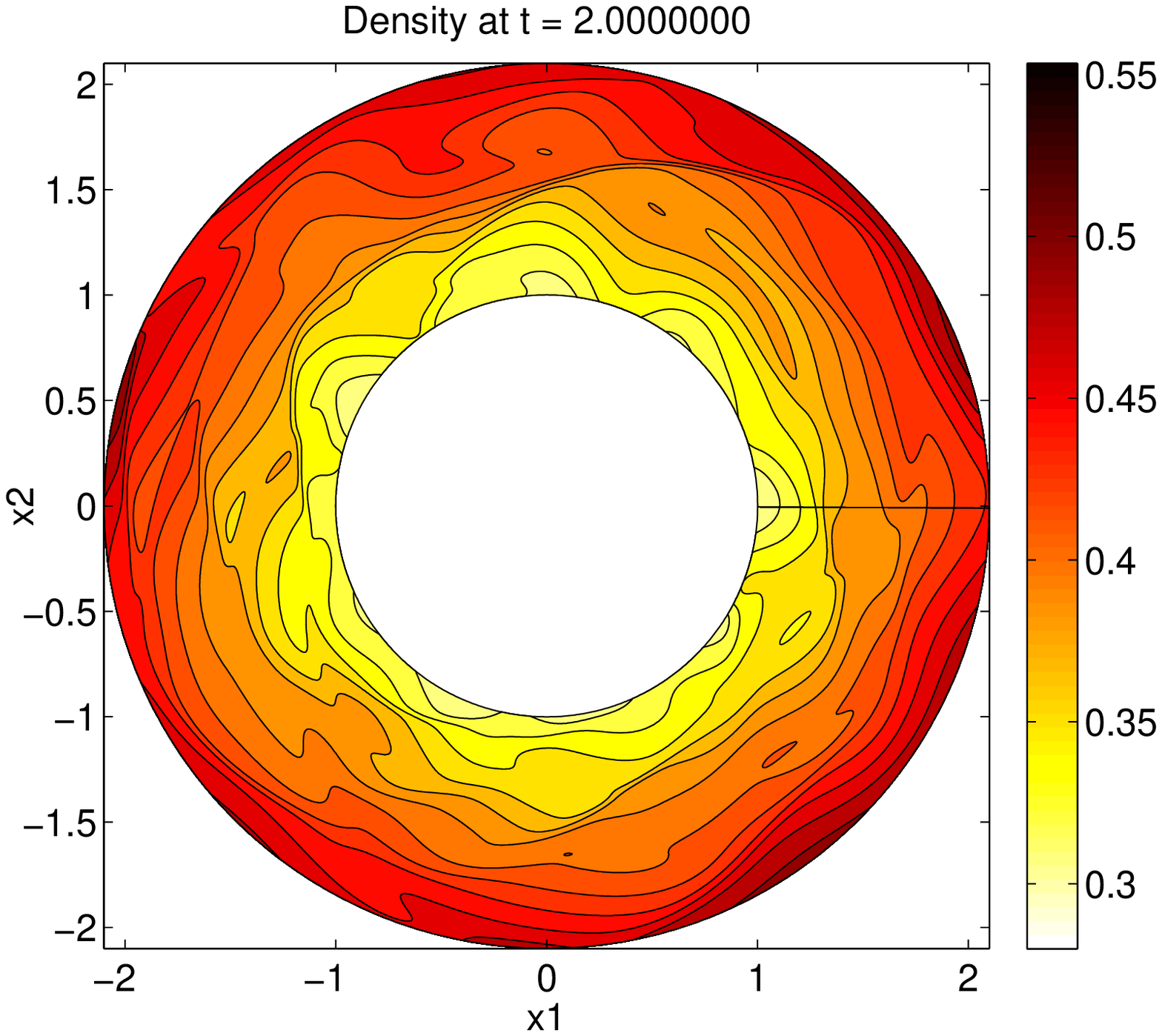}} \\
\end{center}
\caption{
Density $\rho$ as a function of the two-dimensional position coordinates $(x,y)$ in the annulus. The initial condition is given by (\ref{eq:inicondrandom}) with $k_{nm}$ and $\varphi_{nm}$ randomly chosen in the intervals $(0,1]$ and $[0,2\pi]$ respectively according to the uniform distribution, and $\varepsilon = 0.0025$. (a) Linearized solution. (b) Nonlinear solution. For the latter, the numerical simulation is performed with $N_r = N_{\theta} = 640$ and $\Delta t = 0.0005$. The solution is represented at time $t=2$. The density is color coded according to the color bar to the right of the figure. We observe a very good agreement between the linearized and nonlinear solutions, in spite of a fairly large perturbation amplitude.}
\label{fig:random}
\end{figure}

\section{Conclusion}
\label{sec:conclu}
In this paper, we have studied the SOH model on an annular domain. We have linearized the system about perfectly polarized steady-states. and shown that the resulting system has are only pure imaginary eigenvalues and that they form a countable set associated to an ortho-normal basis of  eigenvectors. A numerical scheme for the fully nonlinear system has been proposed. Its results are consistent with the modal analysis for small perturbations  of polarized steady-states. For large perturbations, nonlinear mode-coupling has been shown to result in the progressive turn-on of new modes in a complex fashion. Finally, we have assessed the efficiency of the modal decomposition to analyze the complex patterns of the solution. 
In future work, we will gradually include more physical effects in the model such as adding a repulsive force between the particles to prevent the formation of large concentrations, or immersing the particles in a surrounding fluid to give a better account of the dynamics of active particle suspensions like sperm. Finally, we plan to use the modal analysis to accurately calibrate the model against experimental observations of collective motion.

\end{document}